\documentclass{article}

\usepackage[margin=0.65in]{geometry}

\usepackage{amsmath, amsfonts}
\usepackage[dvips]{graphicx}
\usepackage{theorem}
\usepackage{color}
\usepackage{colortbl}

\newenvironment{proof}[1][Proof]{\textbf{#1.} }{\ \rule{0.5em}{0.5em}}

\def\p{\partial}
\newcommand{~}{\,\,\,\,=\,\,\,\,}
\newtheorem{thm}{Theorem}[section]

\newtheorem{cor}[thm]{Corollary}
\newtheorem{lem}[thm]{Lemma}

\newcommand{\gee}{\,\,\,\,\, \ge \,\,\,\,\,}
\newcommand{\gt}{\,\,\,\,\, > \,\,\,\,\,}
\newcommand{\al}{\alpha}

{\theorembodyfont{\rmfamily}\newtheorem{rem}{Remark}[section]}

\newcommand{\mb}{\mathbf}
\newcommand{\e}{\varepsilon}
\newcommand{\be}{\begin{equation}}

\newcommand{\ee}{\end{equation}}
\newcommand{\bq}{\begin{eqnarray}}
\newcommand{\eq}{\end{eqnarray}}
\newcommand{\ex}{\mathbb{E}}
\newcommand{\half}{\frac{1}{2}}
\newcommand{\bs}{\bigskip}
\newcommand{\nind}{\noindent}
\newcommand{\nn}{\nonumber}
\newcommand{\pdd}[2]{\frac{\partial #1}{\partial #2}}

\newcommand{\lb}{\lbrace}
\newcommand{\rb}{\rbrace}

\newcommand{\sk}{\smallskip}
\newcommand{\mc}{\mathcal}

\newcommand{\pr}{\mathbb{P}}
\usepackage{latexsym}
\usepackage{amssymb}

\begin{document}



\title{\textbf{Small-time asymptotics for basket options \\- the bi-variate SABR model and \\ the hyperbolic heat kernel on $\mathbb{H}^3$}}

\author{
Martin Forde\thanks{Dept. Mathematics, King's College London, Strand, London, WC2R 2LS  ({\tt Martin.Forde@kcl.ac.uk})}\\
\and
Hongzhong Zhang\thanks{Dept. Statistics,
Columbia University,
New York, NY 10027 ({\tt hzhang@stat.columbia.edu})}  }
\maketitle

\begin{abstract}
   We compute a sharp small-time estimate for the price of a basket call under a bi-variate SABR model with both $\beta$ parameters equal to $1$ and three correlation parameters, which extends the work of Bayer,Friz\&Laurence\cite{BFL14} for the multivariate Black-Scholes flat vol model.  The result follows from the heat kernel on hyperbolic space for $n=3$ combined with the Bellaiche\cite{Bel81} heat kernel expansion and Laplace's method, and we give numerical results which corroborate our asymptotic formulae.  Similar to the Black-Scholes case, we find that there is a phase transition from one ``most-likely'' path to two most-likely paths beyond some critical $K^*$. \footnote{We thank Stefano de Marco and the HAJ for many useful discussions, and Peter Friz for alerting us to the updated book version of \cite{BFL14}.}
\end{abstract}

\section{Introduction}

Basket options (i.e. options on a linear combination of $n$ assets) are both interesting and difficult to price, in part because there is no closed-form expression for the price of a basket call option even under a simple $n$-dimensional flat-vol Black-Scholes model for $n \ge 2$, because a sum of independent log-normal random variables is no longer log-normal and does not admit a closed-form density.  Recently, Bayer,Friz\&Laurence\cite{BFL14} compute a small-time estimate for the density associated with pricing a basket call under a two-dimensional Black-Scholes model using Laplace's method; somewhat surprisingly, they find that there is ``phase-transition'' from one ``most-likely'' value for the two stock price values $(S^1_T,S^2_T)$ at maturity $T$, to \textit{two} most-likely values for this vector.

\sk
Gulisashvili\&Tankov\cite{GT15} characterize the implied volatility of a basket
call option at small and large strikes, in a multi-variate
Black-Scholes setting.  \cite{GT15} also compute the leading order term for implied volatility
when the asset prices follow the multidimensional Black-Scholes model evaluated at an independent time-change, and they also deal with a general model where the dependence between assets is described by a copula.  Armstrong et al.\cite{AFLZ14} compute a small-time expansion for implied volatility under a general uncorrelated local-stochastic volatility model using the Bellaiche\cite{Bel81} heat kernel expansion combined with Laplace's method; they also consider the case when the correlation $\rho \le 0$ and in this case the approach still works if the drift of the volatility takes a specific functional form and there is no local volatility component, which includes the SABR model for $\beta=1, \rho \le 0$.

\sk
   In this article, we compute the small-time behaviour of a basket call option under a \textit{bi-variate} uncorrelated SABR model with $\beta$ parameters equal to $1$ with three correlation parameters for two correlated assets, and (as for the Black-Scholes case) we find that the same phase transition effect occurs at the same critical strike, i.e. the ``rate function'' in the exponent of the saddlepoint approximation is qualitatively different for strikes values greater than some critical strike $K^*$, where there is a phase transition from one ``most-likely'' configuration to two most-likely paths.  The result follows from the known expression for the hyperbolic heat kernel on the upper half plane $\mathbb{H}^3$ combined with the Bellaiche small-time heat kernel expansion over compact domains and Laplace's method (in a similar spirit to \cite{AFLZ14}).  We then extend this result to a correlated bivariate SABR model with three correlation parameters.

\sk

\section{Background on the heat kernel and hyperbolic space}

\sk

Consider a diffusion process on $M=\mathbb{R}^n$ with infinitesimal generator $L$.  In local coordinates, $L$ takes the form
\bq
\label{eq:SK}
L ~\frac{1}{2}\sum_{1 \le i,j\le n}a^{ij}(x)\pdd{^2}{ x_i \partial x_j} \,+ \, \sum_{1 \le i \le n} b^i(x) \pdd{}{x_i}\,.\nn
\eq
\nind (see Theorem \ref{thm:Bellaiche} below for the conditions that we impose on $b,\sigma$ via conditions on $\mc{M}$ and $\mc{A}$).  Now furnish $M$ with a Riemmanian metric $g_{ij}=(a_{ij})^{-1}$ so that $M$ is a smooth Riemmanian manifold with a single chart given by the identity map.  We can write $L$ as $\frac{1}{2}\Delta+\mathcal{A}$, where $\Delta=\sum_{i,j}\frac{1}{\sqrt{|g|}}\partial_i (\sqrt{|g|}\,g^{ij}\partial_j)$ is the Laplace-Beltrami operator and $\mathcal{A}^i=b^i-\frac{1}{2}\sum_j \frac{1}{\sqrt{|g|}}\partial_j (\sqrt{|g|}\,g^{ij})$ is a smooth first-order differential operator and $|g|=|\det g_{ij}|$ (recall that $g^{ij}=(g_{ij})^{-1}$).

\bs

Given such an operator $L$, the heat kernel $p(\mb{x};\mb{y},t)$ of $L$ is the fundamental solution to the heat equation $\p_t u =(\mc{A}+\half \Delta) u$, which is also the transition density of the diffusion $X$ with respect to the Riemannian volume measure $\sqrt{|g|}$ (see \cite{Hsu} for more details); to obtain the transition density $\hat{p}(\mb{x};\mb{y},t)$ of $X$ with respect to Lebesgue measure $dx_1 ... dx_n$, we set
\bq
\hat{p}(\mb{x};\mb{y},t)&=&p(\mb{x};\mb{y},t)\sqrt{g(\mb{y})}\, \nn
\eq
where $\sqrt{g}$ is shorthand for $\sqrt{\mathrm{det}\,g}$.

\sk
\sk

Throughout, we let $\rho(\mb{x},\mb{y})$ denote the Riemannian distance between two points $\mb{x},\mb{y}\in M$.

\sk
\sk

\begin{thm} \label{thm:Bellaiche} (Theorem 4.1 in \cite{Bel81}).  Let $M$ be a $C^4$-Riemannian manifold and $\mathcal{A}$ a $C^4$-vector field. The heat kernel $p(\mb{x};\mb{y},t)$ of the operator $\frac{1}{2}\Delta + \mathcal{A}$ satisfies:
\bq \label{eq:MolchanovFormula}
 p(\mb{x};\mb{y},t) &=& (2\pi t)^{-\frac{n}{2}} u_0(\mb{x},\mb{y}) \,e^{-\half \rho(\mb{x},\mb{y})^2/t+A(\mb{x},\mb{y})}[1+o(1)] \quad \quad \quad \quad (t \to 0)
\eq
for some function $u_0(\mb{x},\mb{y})$ (see e.g. \cite{Hsu02} or \cite{AFLZ14} for details on how to compute $u_0(\mb{x},\mb{y})$), where $(\mb{x},\mb{y}) \in (M \times M) \setminus C(M)$ \footnote{$C(M)$ is the subset of points $(\mb{x},\mb{y})$ in $M \times M$ such that $x$ lies in the cut locus of $y$}, and
\bq
A(\mb{x},\mb{y}) &=&\int_0^1 \langle \mathcal{A},\dot{\gamma}(s) \rangle \, ds \nn
\eq
\nind for the unique distance-minimizing geodesic $\gamma:[0,1]\longrightarrow M$ joining $\mb{x}$ and $\mb{y}$. The estimate \eqref{eq:MolchanovFormula} is uniform on compact subsets of $(M \times M) \setminus C(M)$.
\end{thm}

We know that Theorem 1.1 holds in general.  However, when $M$ is the upper half space $\mathbb{H}^3=\lb (x,y,a) \in \mathbb{R}^3 : a>0 \rb$, and the metric $(g_{ij})$ is the Poincar\'{e} metric with line element $ds^2=\frac{1}{a^2}(dx^2+dy^2+da^2)$ and $\mc{A}=0$, from e.g. \cite{GM98} we also have the known exact formula:
\bq
\label{H3ExactSolN}
p(\mb{x};\mb{y},t)&=&  \frac{1}{(2 \pi t)^{\frac{3}{2}}}\,\frac{\rho}{\sinh \rho} \, e^{-\half t-\frac{\rho^2}{2t}}  \,,
\eq
where $\rho$ is shorthand for $\rho(\mb{x},\mb{y})$.  Equating \eqref{H3ExactSolN} with \eqref{eq:MolchanovFormula} for  $\mc{A}=0$, we see that we must have that
\bq
u_0(\mb{x},\mb{y}) &=&\frac{\rho}{\sinh \rho} \nn \,
\eq
for this metric.

\sk
\sk

\section{The bi-variate SABR model - zero correlation case}

\bs

We work on a model $(\Omega,\mc{F},\mathbb{P})$ throughout, with a filtration $\mc{F}_t$ supporting three independent Brownian motions which satisfies the usual
conditions.

\sk
\sk

We now consider the following bi-variate SABR model for two asset price processes $S^{(1)}_t,S^{(2)}_t$:
\bq
\label{eq:Model}
        \left\{
        \begin{array}{ll}
dX_t\,=\,\, -\half a_t^2 dt + a_t dW^1_t \,, \\
dY_t \,\,=\,\, -\half a_t^2  dt + a_t dW^2_t \,,\\
da_t \,\,\,= \,\, a_t dW^3_t \,
        \end{array}\
        \right.
\eq
where $X_t=\log S^{(1)}_t,Y_t=\log S^{(2)}_t$ and $W^1,W^2,W^3$ are three independent standard Brownian motions and $a_0>0$.  The law of $(X_t-X_0,Y_t-Y_0)$ is independent of $(X_0,Y_0)$, so without loss of generality we set $X_0=0,Y_0=0$.

\bs

\nind We first recall some facts about the geometry associated with this model:

\begin{itemize}
\item  The associated Riemmanian metric $(g_{ij})$ is the three-dimensional hyperbolic metric on $\mathbb{R}\times \mathbb{R} \times \mathbb{R}^+$ with line element $ds^2=\frac{1}{a^2}(dx^2+dy^2+da^2)$, and volume element $\sqrt{g}=\frac{1}{a^3}$.
\item From e.g. page 179 in \cite{HL08}, the geodesic distance between two points $x_0,y_0,a_0$ and $(x_0,y_0,a_0)$ is given by
\bq
\label{eq:dist}
\rho ~ \rho(x_0,y_0,a_0;x,y,a) &=& \cosh^{-1}[1+\frac{|(x-x_0)^2 +(y-y_0)^2 +(a-a_0)^2|}{2 a_0 a}]  \,.
\eq

\item The straight lines perpendicular to $a=0$ and the circles of $\mathbb{H}^3$ whose planes are perpendicular to the hyperplane $a=0$ and whose centres are in this hyperplane are the geodesics of $\mathbb{H}^3$ (see Proposition 3.1 on page 127 in doCarmo\cite{doC92}).

\item The Laplace-Beltrami operator for $\mathbb{H}^3$ is given by
\bq
\Delta &=&
 a^2 (\p_x^2 +\p_y^2+\p_a^2) \,-\,a \p_a \nn \,
\eq
(see e.g. Eq 3.2 in \cite{MY05}).
\item For this model, $\mc{A}$ is given by $\mathcal{A}^i=b^i-\frac{1}{2}\sum_j \frac{1}{\sqrt{g}}\partial_j (\sqrt{g}\,g^{ij})$ so $\mc{A}=(-\half a^2, -\half a^2,\half a)$.  Then we have
\bq
\label{eq:mathcalA}
A(x_0,y_0,a_0;x,y,a) &=&  \int_0^1 \langle \mathcal{A},\dot{\gamma}\rangle \,dt \nn \\
&=& \int_{\gamma} [\frac{1}{a^2}\mathcal{A}^1  \frac{dx}{dt}  \,+ \, \frac{1}{a^2}\mathcal{A}^2 \frac{dy}{dt}  + \frac{1}{a^2}\mathcal{A}^3 \frac{da}{dt}] dt \nn \\
&=& \int_{\gamma} [-\half \frac{dx}{dt}  \,-\, \half \frac{dy}{dt}  + \frac{1}{2 a} \frac{da}{dt}] dt \nn \\
&=& -\half (x-x_0) -\half (y-y_0) \,+\, \half \log \frac{a}{a_0}  \,.
\eq

\item 
Combining \eqref{eq:MolchanovFormula}, \eqref{H3ExactSolN} and \eqref{eq:mathcalA}, we see that the density of $(X_t,Y_t,a_t)$ has the following small-time behaviour over any compact set of $(x,y,a)$:
\bq \label{eq:MolchanovFormula2}
 \hat{p}(x_0,y_0,a_0;x,y,a,t) ~  \hat{p}(x,y,a,t) &=& \sqrt{g}\,\,e^{-\half (x-x_0+y-y_0) \,+\, \half \log \frac{a}{a_0}}\,\frac{1}{(2 \pi t)^{\frac{3}{2}}}\frac{\rho}{\sinh \rho} \, e^{-\frac{\rho^2}{2t}}[1+o(1)]   \nn\\
  &=& \frac{\sqrt{a}}{\sqrt{a_0}} \frac{1}{a^3}\,e^{-\half (x-x_0+y-y_0)}\,\frac{1}{(2 \pi t)^{\frac{3}{2}}}\frac{\rho}{\sinh \rho} \, e^{-\frac{\rho^2}{2t}}[1+o(1)]  \nn \\
  &=& \frac{1}{\sqrt{a_0}a^{\frac{5}{2}}} \,e^{-\half (x-x_0+y-y_0)}\,\frac{1}{(2 \pi t)^{\frac{3}{2}}}\frac{\rho}{\sinh \rho} \, e^{-\frac{\rho^2}{2t}}[1+o(1)] \quad \quad \quad \quad (t \to 0) \, .\nn \\
\eq

\item
Temporarily switching variables, we know that $\hat{p}(x,y,a;x_1,y_1,a_1,t)$ is a solution to the backward heat equation
\bq
\p_t \hat{p} &=&-\half a^2 ( \hat{p}_x+ \hat{p}_y) \,+\, \half a^2 ( \hat{p}_{xx} + \hat{p}_{yy}+ \hat{p}_{aa})   \nn \,
\eq
subject to $\hat{p}(x,y,a;x_1,y_1,a_1,t)=\delta(x-x_1,y-y_1,a-a_1)$.  If we now let
\bq
\hat{p}(x,y,a;x_1,y_1,a_1,t) &= & e^{\half (x-x_1+y-y_1) \,+\, \half \log \frac{a_1}{a}} q(x,y,a;x_1,y_1,a_1,t)\eq then the PDE transforms to
\bq
\label{eq:qPDE}
\p_t q &=&-\half a q_a\,+\, \half a^2 ( q_{xx} + q_{yy}+ q_{aa})\,+\,V(y) q   ~ \half \Delta  q \,+\, V(y) q   \,
\eq
\nind with $q(x,y,a;x_1,y_1,a_1,t)=\delta(x-x_1,y-y_1,a-a_1)$, where $V(y)=\frac{3}{8}-\frac{1}{4}a^2$ and $\Delta$ is the Laplace-Beltrami operator as before.

\sk
\sk

\item The law of $(X_t-x_0,Y_t-y_0)$ is independent of $x_0,y_0$, so without loss of generality we set $x_0=y_0=0$ from here on.  Then from \eqref{eq:qPDE} we see that
\bq
\label{eq:GlobalBound}
\hat{p}(x_0,y_0,a_0;x,y,a,t) &=&  e^{-\half(x+y)} \,\frac{1}{a^3} \frac{\sqrt{a}}{\sqrt{a_0}}\ e^{\frac{3}{8}t}
\,\frac{1}{dx dy da}\mathbb{E}^{\mathbb{P}^0}(e^{-\frac{1}{4}\int_0^t a_s^2 ds} \,1_{(X_t,Y_t,a_t)\in (dx,dy,da)})\nn \\
 &\le &  e^{-\half(x+y)} \,\frac{1}{a^3}  \frac{\sqrt{a}}{\sqrt{a_0}}\ e^{\frac{3}{8}t}
\,\frac{1}{dx dy da}\mathbb{E}^{\mathbb{P}^0}( \,1_{(X_t,Y_t,a_t)\in (dx,dy,da)})\nn \\
&= & e^{-\half(x+y)} \frac{1}{a^3}  \, \frac{\sqrt{a}}{\sqrt{a_0}}\ e^{\frac{3}{8}t}\, \hat{p}^0(x,y,a,t)\nn \\
&=& e^{-\half(x+y)} \, \frac{1}{\sqrt{a_0} a^{\frac{5}{2}}}\ e^{\frac{3}{8}t}\frac{1}{(2 \pi t)^{\frac{3}{2}}}\,\frac{\rho}{\sinh \rho} \, e^{-\half t-\frac{\rho^2}{2t}}\nn\\
&=&e^{-\half(x+y)} \, \frac{1}{\sqrt{a_0} a^{\frac{5}{2}}}\ e^{-\frac{1}{8}t}\frac{1}{(2 \pi t)^{\frac{3}{2}}}\,\frac{\rho}{\sinh \rho} \, e^{-\frac{\rho^2}{2t}}
\eq
where $\mathbb{P}^0$ and $\hat{p}^0$ denote the measure and transition density associated with the Laplace-Beltrami operator i.e. without the additional $\mc{A}$ term.
\eqref{eq:GlobalBound} provides a \textit{global} upper bound on the transition density for $(X_t,Y_t)$ we can use to deal with the tail integrals outside the compact set where we are applying the Bellaiche heat kernel
expansion.
\end{itemize}
From here on (in contrast to the previous section) we work in log space (i.e. $x$ and $y$ will refer to the log of the first and second asset price process),
which will be more convenient when working with the hyperbolic metric and the heat kernel.  We also introduce the following quantities which will be needed in Theorem \ref{thm:SABRrhozero}.
\bq
x^*(k)&:=&\mathrm{argmin}_{x \le \log(\half K)}\bar{H}_K(x) \nn \\
\bar{H}_K(x)&:=&\cosh^{-1}\sqrt{1+ \frac{1}{a_0^2}[\,x^2+(\log(K-e^x))^2]}\nn \\ \varphi(k)&:=&\frac{1}{2}[\bar{H}_K(x^*(k))]^2\eq
and
\bq
a^*(x,y)&:=&\sqrt{a_0^2 + x^2  + y^2 } \nn \,\nn \\
y^*&:=&y^*(k)~\log(K-x^*) \,,\nn \\
\Phi(x,y,a)&:=&\half \rho(a_0,0,0;a,x,y)^2 \,,\nn \\
\Psi(x)&:=&\half\bar{H}_K(x)^2.
\eq

\begin{thm}
\label{thm:SABRrhozero} For the uncorrelated model in \eqref{eq:Model} with $X_0=0,Y_0=0$, we have the following small-time behaviour for basket call options for $K \in (2,\infty)$ with $K \ne K^*:=2 e$:
\bq
\label{eq:SABRCallEstimate}
\mathbb{E}(S^{(1)}_t+S^{(2)}_t-K)^+ &=&
 \psi(k)\, t^{\frac{3}{2}}e^{-\frac{\varphi(k)}{t}}[1+o(1)] \quad \quad \quad \quad (t \to 0)
\eq
where $k=\log K$,and
\bq
\psi(k) &=&\,(1+1_{k>k^*})\,\cdot \frac{e^{-y^*} (e^{2x^*}+e^{2y^*}) }{\sqrt{a_0 a^*(x^*,y^*)}}\,\,\frac{e^{-\half (x^* +y^*)}}{ \sqrt{2\pi \Phi_{aa}(x^*,y^*,a^*(x^*,y^*))\,\Psi''(x^*)} } \,  \frac{1}{\bar{H}_K(x^*)\sinh \bar{H}_K(x^*)} \nn \,
\eq
When $K\in (2,K^*]$, $x^*$ and $\varphi$ simplify to $x^*(K)=\log(\half K)$ and $\varphi(k)=\half[\cosh^{-1}(\sqrt{1+\frac{2}{a_0^2}(k-\log 2)^2})]^2$.
For $K=2 e$, we have the special behaviour:
\bq
\mathbb{E}(S^{(1)}_t+S^{(2)}_t-K)^+ &=&  \frac{1}{\sqrt{a_0 \bar{a} }}\,\,\frac{1}{4 \pi \sqrt{\Phi_{aa}(x^*,y^*,\bar{a})} } \,\frac{\Gamma(\frac{1}{4})\, }{(\xi \bar{H})^{\frac{1}{4}} \, } \, \frac{1}{\bar{H} \sinh\bar{H}} \,t^{\frac{5}{4}} e^{-\frac{\bar{H}^2}{2t}}[1+o(1)] \nn
\eq
where $\bar{H}=\cosh^{-1}\sqrt{1+\frac{2}{a_0^2}}$, $\bar{a}=\sqrt{1+\frac{2}{a_0^2}}$ and $\xi=\frac{5}{12}\frac{1}{\sqrt{2a_0^2 +4}}$.
\end{thm}

\begin{figure}
\centering
\includegraphics[width=140pt,height=140pt]{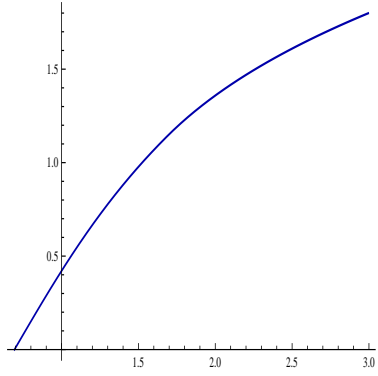}
\nind \caption{Here we have plotted $\varphi(k)$ for the uncorrelated SABR model in \eqref{eq:Model}.}
\end{figure}

\nind \begin{proof}
We break the proof into several parts.
\begin{itemize}
\item
\textbf{Computing the small-time behaviour of $\mathbb{E}(a_t^2 \delta(X_t-x,Y_t-y))$}
\sk

 From the generalized It\^{o} formula, we have
\bq
\label{eq:TanakaMeyer}
\mathbb{E}(S^{(1)}_t+S^{(2)}_t-K)^{+}\,-\,\mathbb{E}(2S_0-K)^{+}
&=& \frac{1}{2} \int_0^t  \mathbb{E}( a_u^2 \,\delta(S^{(1)}_u+S^{(2)}_u-K) [(S^{(1)}_u)^2+(S^{(2)}_u)^2]) du
\eq
\nind To this end, we first compute the small-time behaviour of  $\mathbb{E}(a_t^2 \delta(X_t-x,Y_t-y)) = \int_{a=0}^{\infty} a^2  \hat{p}(x,y,a,t) da$.
 Fix a sufficiently large constant $M>0$, let $\e=\half\min\{a_0,1/\sqrt{a_0^2+2M^2}\}$.
  Using the heat kernel expansion in \eqref{eq:MolchanovFormula2} over a compact domain $(x,y,a)\in[-M,M]\times[-M,M]\times[\e,1/\e]$ in \eqref{eq:MolchanovFormula2}, and Laplace's method (see \cite{SS03}), we find that
 \bq
  \mathcal{I}_1&:=&\int_\e^{1/\e}a^2  \hat{p}(x,y,a,t) da\nn\\
  &=&\frac{(a^*(x,y))^2}{\sqrt{a_0}\, (a^*(x,y))^\frac{5}{2}}\,e^{-\half (x+y)}\,\frac{1}{2 \pi t \sqrt{\Phi_{aa}(x,y,a^*(x,y))}} \, \frac{\rho^*(x,y)}{\sinh \rho^*(x,y)} \, e^{-\frac{\rho^*(x,y)^2}{2t}}[1+o(1)] \nn \\
 &=&  \frac{1}{\sqrt{a_0 a^*(x,y)}}\,e^{-\half (x+y)}\,\frac{1}{2 \pi t \sqrt{\Phi_{aa}(x,y,a^*(x,y))}} \, \frac{\rho^*(x,y)}{\sinh \rho^*(x,y)} \, e^{-\frac{\rho^*(x,y)^2}{2t}}[1+o(1)]\label{I1}
\eq
and
\bq
a^*(x,y) ~ \sqrt{a_0^2 + x^2  + y^2 } \,\quad , \quad \rho^*(x,y) ~ \cosh^{-1} [\sqrt{1 + (x^2  + y^2)/a_0^2}\,]\nn
\eq
\nind are the minimizer (resp. minimum) of $\rho(a_0,0,0;a,x,y)$ over all $a \in \mathbb{R}^+$.  Moreover, the function $\rho(a_0,0,0;a,x,y)$ is strictly decreasing in $a$ over $(0,a^*(x,y)]$ and is strictly increasing in $a$ over $[a^*(x,y),\infty)$. On the other hand, using the global bound \eqref{eq:GlobalBound} we have that
\bq
0\le\mathcal{I}_2&:=&\int_0^{\e}a^2  \hat{p}(x,y,a,t) da\nn\\
&\le&e^{-\half(x+y)}\frac{e^{-\frac{1}{8}t}}{\sqrt{a_0}}\frac{1}{(2\pi t)^{\frac{3}{2}}}\int_0^{\e}
\frac{1}{\sqrt{a}}\frac{\rho}{\sinh\rho}e^{-\frac{\rho^2}{2t}}da\nn\\
&\le&e^{-\half(x+y)}\frac{e^{-\frac{1}{8}t}}{\sqrt{a_0}}\frac{1}{(2\pi t)^{\frac{3}{2}}}\int_0^{\e}
\frac{1}{\sqrt{a}}\frac{\rho(a_0,0,0;\e,x,y)}{\sinh\rho(a_0,0,0;\e,x,y)}e^{-\frac{\rho(a_0,0,0;\e,x,y)^2}{2t}}da\nn\\
&\le&2\sqrt{\e}\, e^{-\half(x+y)}\frac{1}{\sqrt{a_0}}\frac{1}{(2\pi t)^{\frac{3}{2}}}\frac{\rho(a_0,0,0;\e,x,y)}{\sinh\rho(a_0,0,0;\e,x,y)}e^{-\frac{\rho(a_0,0,0;\e,x,y)^2}{2t}},\label{I2}
\eq
where the second line follows from the fact that both $\frac{\rho}{\sinh\rho}$ and $e^{-\frac{\rho^2}{2t}}$ are positive, decreasing functions of $\rho>0$, and that $\rho(a_0,0,0;a,x,y)$ is strictly decreasing for $a\in(0,\e]$. Similarly, $\rho \sim \log a$ as $a\to\infty$ so $\frac{\rho}{\sinh\rho}=o(e^{-\frac{3}{4}\rho})=o(a^{-\frac{3}{4}})$ as $a\to\infty$. Hence there exists a constant $C>0$ such that $0\le \frac{\rho}{\sinh\rho}\le Ca^{-\frac{3}{4}}$ for all $a>\frac{1}{\e}$ and $(x,y)\in[-M,M]\times[-M,M]$. It follows that
\bq
0\le\mathcal{I}_3&:=&\int_{1/\e}^\infty a^2  \hat{p}(x,y,a,t) da\nn\\
&\le&e^{-\half(x+y)}\frac{e^{-\frac{1}{8}t}}{\sqrt{a_0}}\frac{1}{(2\pi t)^{\frac{3}{2}}}\int_{1/\e}^\infty
\frac{1}{\sqrt{a}}\frac{\rho}{\sinh\rho}e^{-\frac{\rho^2}{2t}}da\nn\\
&\le&e^{-\half(x+y)}\frac{e^{-\frac{1}{8}t}}{\sqrt{a_0}}\frac{C}{(2\pi t)^{\frac{3}{2}}}\int_{1/\e}^\infty
a^{-\frac{5}{4}}e^{-\frac{\rho(a_0,0,0;1/\e,x,y)^2}{2t}}da\nn\\
&\le&4 \,\e^{\frac{1}{4}}e^{-\half(x+y)}\frac{C}{\sqrt{a_0}}\frac{1}{(2\pi t)^{\frac{3}{2}}}\,e^{-\frac{\rho(a_0,0,0;1/\e,x,y)^2}{2t}}.\label{I3}
\eq
Since $\rho^*(x,y)<\min\{\rho(a_0,0,0;\e,x,y),\rho(a_0,0,0;1/\e,x,y)\}$, we know that $\mathcal{I}_i=o(\mathcal{I}_1)$ for $i=2,3$, as $t\to0$. Thus, for any $(x,y)\in[-M,M]\times[-M,M]$, we have
\be
 \mathbb{E}(a_t^2 \delta(X_t-x,Y_t-y)) =   \frac{1}{\sqrt{a_0 a^*(x,y)}}\,e^{-\half (x+y)}\,\frac{1}{2 \pi t \sqrt{\Phi_{aa}(x,y,a^*(x,y))}} \, \frac{\rho^*(x,y)}{\sinh \rho^*(x,y)} \, e^{-\frac{\rho^*(x,y)^2}{2t}}[1+o(1)]\quad\quad\quad (t\to0) \quad \quad \label{J1}
\ee
and this expansion is uniform for all $(x,y)\in[-M,M]\times[-M,M]$.
\sk
%
\sk
\sk
On the other hand, by standard properties of the Dirac delta function, we also have
\bq
\mathbb{E}(a_t^2 \delta(S^{(1)}_t-e^x,S^{(2)}_t-e^y)) &=& \frac{1}{e^x e^y}\mathbb{E}(a_t^2 \delta(X_t-x,Y_t-y))
\label{eq:Dirac} \,.
\eq
\item \textbf{Computing the convolution}

\sk
\nind Applying a convolution to this we see that
\bq
\mc{J}&:=&\mathbb{E}(a_t^2 \delta(S^{(1)}_t+S^{(2)}_t-K) [(S^{(1)}_t)^2+(S^{(2)}_t)^2] )\nn \\
&=&\int_{e^x=0}^K\mathbb{E}(a_t^2 \delta(S_t^{(1)}-e^x,S_t^{(2)}-e^{y_K(x)}) [e^{2x}+e^{2y_K(x)}]d(e^x)\nn\\
&=&\int_{-\infty}^{\log K}\frac{[e^{2x}+e^{2y_K(x)}]}{e^{y_K(x)}}\mathbb{E}(a_t^2\delta(X_t-x,Y_t-y_K(x))dx\label{J3}
\eq
where $y_K(x):=\log(K-e^x)$.  Now fix $R>0$ sufficiently large,
we may use the estimate in \eqref{J1} for the compact domain $\{(x,y):x\in[-R,\log(K-e^{-R})], e^x+e^y=K\}$, we have that
\bq
\mc{J}_1&:=&\int_{-R}^{\log(K-e^{-R})}\frac{[e^{2x}+e^{2y_K(x)}]}{e^{y_K(x)}}\mathbb{E}(a_t^2\delta(X_t-x,Y_t-y_K(x)))dx\nn\\
&=& [1+o(1)]\int_{-R}^{\log(K-e^{-R})}  \frac{ [e^{2x}+e^{2y_K(x)}]}{e^{y_K(x)}} \frac{e^{-\half (x+y_K(x))}}{\sqrt{a_0 a^*(x,y_K(x))}}\,\frac{\bar{H}_K(x)}{\sinh \bar{H}_K(x)} \frac{ e^{-\frac{\bar{H}_K(x)^2}{2t}}\, dx}{2 \pi t   \sqrt{\Phi_{aa}(x,y_K(x),a^*(x,y_K(x)))} }  \nn
\eq

\nind as $t\to0$, where
\bq
\bar{H}_K(x)&=&\rho^*(x,y_K(x)) ~ \cosh^{-1}\sqrt{1+ [\,x^2+(\log(K-e^x))^2]/a_0^2}  \,, \nn \,. \nn
\eq
Note that we can choose $R>0$ sufficiently large so that $\bar{H}_K(x)$ is strictly decreasing over $(-\infty,-R)$, and that $\bar{H}_K(x)$ is strictly increasing over $(\log(K-e^{-R}),\log K)$.
To apply Laplace's method for $\mc{J}_1$, we have to minimize $\bar{H}_K(x)$ over all $x$ in the allowable range $[-R,\log(K-e^{-R})]$. By the monotonicity of $\bar{H}_K(x)$, the minimum over $[-R,\log(K-e^{-R})]$ is the same as $\inf_{x<\log K}\bar{H}_K(x)$. Thus we have to minimize $x^2+\log(K-e^x)$ over $x$ in this range, for which we will need the following Lemma:

\sk
\begin{lem}
\label{lem:BFLError}
Let $\bar{h}_K(z)=(\log z)^2+(\log (K-z))^2$ and set $K^*=2e$.  Then we have the following classification for the minimizer(s) of $\bar{h}_K(z)$:
\begin{itemize}
\item If $K\in(0,K^*)$, $\bar{h}_K''(z)=\bar{h}_K''(K-z)>0$ for all $z\in(0,K)$, so $z^*=\half K$ is the unique global minimum of $\bar{h}_K$;
\item If $K=K^*$, $\bar{h}_K''(z)=\bar{h}_K''(K-z)>0$ for all $z\in(0,K)\backslash\{\half K\}$ and $\bar{h}_K''(\half K)=\bar{h}_K'''(\half K)=0$ and $\bar{h}_K''''(\half K)>0$, so $z^*=\half K$ is the unique global minimum of $\bar{h}_K$;

\item If $K\in(K^*,\infty)$, $\bar{h}_K$ has two global minima, $x^*$ and $K-x^*$, with $x^*\in(0,\half K)$.
\end{itemize}
\end{lem}
\nind \begin{proof}[Proof of Lemma \ref{lem:BFLError}]
See Appendix A (see also page 3 in \cite{BFL14} for a statement of this result and a plot of the three different cases).
\end{proof} \\

 Let us denote by
\be
\Psi(x):=\frac{1}{2}\bar{H}_K(x)^2  \quad, \quad  \varphi(k):=\inf_{x<k}\Psi(x)>0,\,\,\,\text{where}\,\,\,k=\log K.
\ee

\item \textbf{$K \in (0,2e)$}

\sk
Applying Laplace's method to $\mc{J}_1$ for $K \in (0,2e)$, we see that
\bq
\label{eq:JayOne}
\mc{J}_1&=&  \,\frac{e^{-y^*} (e^{2x^*}+e^{2y^*})}{\sqrt{a_0 a^*(x^*,y^*)}}\,\frac{e^{-\half (x^* +y^*)}}{\sqrt{2 \pi t} \sqrt{\Phi_{aa}(x^*,y^*,a^*(x^*,y^*))} } \,\frac{1}{\sqrt{\Psi''(x^*)}} \, \frac{\bar{H}_K(x^*)}{\sinh \bar{H}_K(x^*)} \, e^{-\frac{\varphi(k)}{t}}[1+o(1)] \,,\nn \\
\eq
where $x^*\equiv x^*(K)=\log(\half K)$
and $y^*=y_K(x^*)$, and $\varphi(k)$ can be calculated explicitly as
 \bq
\varphi(k)&=&\half[\cosh^{-1}(\sqrt{1+\frac{2}{a_0^2}(k-\log 2)^2})\,]^2 \nn
 \eq
and we can re-write $x^*$ as $x^*=k-\log 2$.  By trivial adjustment to \eqref{eq:JayOne} (using \eqref{eq:Dirac}) we also see that the exact density $f(K)$ of $S^1_t+S^2_t$ has the asymptotic behaviour
\bq
\label{eq:SABRDensityApproxRhoZero}
f(K) &=& e^{-y^*}  \,\frac{1}{\sqrt{a_0} a^*(x^*,y^*)^{\frac{5}{2}}}\,\frac{1}{\sqrt{2 \pi t} \sqrt{\Phi_{aa}(x^*,y^*,a^*(x^*,y^*))} } \,\frac{1}{\sqrt{\Psi''(x^*)}} \, \frac{\bar{H}_K(x^*)}{\sinh \bar{H}_K(x^*)} \, e^{-\frac{\varphi(k)}{t}}[1+o(1)] \,.\nn \\
\eq

\item \textbf{$K>2e$}

\sk

 Similarly, for $K>2e$ we have
\bq
\label{eq:JayOneDouble}
 \mc{J}_1&=&  \,\frac{2e^{-y^*} (e^{2x^*}+e^{2y^*})}{\sqrt{a_0 a^*(x^*,y^*)}}\,\frac{e^{-\half (x^* +y^*)}}{\sqrt{2 \pi t} \sqrt{\Phi_{aa}(x^*,y^*,a^*(x^*,y^*))} } \,\frac{1}{\sqrt{\Psi''(x^*)}} \, \frac{\bar{H}_K(x^*)}{\sinh \bar{H}_K(x^*)} \, e^{-\frac{\varphi(k)}{t}}[1+o(1)] \,,
\eq
where now $x^*=\mathrm{argmin}_{x \le \log K}\bar{H}_K(x)=\mathrm{argmin}_{x \le \log(\half K)}\bar{H}_K(x)$ (this equality follows from the second bullet point in Lemma \ref{lem:BFLError}), $y^*=y_K(x^*)$, and note that we also write this as $\log \mathrm{argmin}_{x \in (0, \half K]}\bar{h}_K(x)$ where $\bar{h}_K(\cdot)$ is defined as in Lemma \ref{lem:BFLError}.  Note that \eqref{eq:JayOneDouble} is twice the expression in \eqref{eq:JayOne} because we now have two saddlepoints.

 \bs
\item \textbf{The special value $K^*=2e$}

\sk
 At the special value $K^*=2e$, the second derivative term $\bar{H}_K''(x^*)$ vanishes, and we have
\bq
\bar{H}_{K^*}(x) &=&\bar{H}_{K^*}(x^*)\,+\, \xi(x-x^*)^4 \,+\,O((x-x^*)^5) \nn \,,
\eq
so
\bq
\bar{H}_{K^*}(x)^2 &=&\bar{H}_{K^*}(x^*)^2+ 2\xi \bar{H}_{K^*}(x^*) (x-x^*)^4 +O((x-x^*)^5) \nn \,,
\eq
where $x^*\equiv x^*(K^*)=1$ (and hence $y^*=y_K(x^*)=1$), and $\xi=\frac{5}{24\sqrt{1+a_0^2/2}}$.  Using the identity $
\int_{-\infty}^{\infty} e^{-\zeta x^4} dx = \frac{\Gamma(\frac{1}{4})}{2 \zeta^{\frac{1}{4}}}$, we now obtain
\bq
\mc{J}_1&=& \ \frac{e^2+e^2}{e^1}e^{-\half (1+1)}\frac{1 }{\sqrt{a_0 \bar{a} }}\,\,\frac{1}{2 \pi t\sqrt{\Phi_{aa}(x^*,y^*,\bar{a})} } \,\frac{\Gamma(\frac{1}{4})\, }{2(\xi  \bar{H}_{K^*}(x^*))^{\frac{1}{4}} \, t^{-\frac{1}{4}}} \, \frac{\bar{H}}{\sinh \bar{H}} \, e^{-\frac{\bar{H}^2}{2t}}[1+o(1)] \nn \\
&=& \,\frac{2}{\sqrt{a_0 \bar{a} }}\,\,\frac{1}{4 \pi \sqrt{\Phi_{aa}(x^*,y^*,\bar{a})} } \,\frac{\Gamma(\frac{1}{4})\, }{(\xi \bar{H})^{\frac{1}{4}} \, t^{\frac{3}{4}}} \, \frac{\bar{H}}{\sinh \bar{H}} \, e^{-\frac{\bar{H}^2}{2t}}[1+o(1)]
\label{eq:J1Special}
\eq
where $\bar{H}=\bar{H}_{K^*}(x^*)=\cosh^{-1}\sqrt{1+\frac{2}{a_0^2}}$, $\bar{a}=a^*(x^*,y^*)=\sqrt{1+\frac{2}{a_0^2}}$.


\sk
\sk

\item \textbf{Controlling the tail integrals}

\sk

We now control the tail integrals in \eqref{J3}. To this end, denote
$
T_t:=\int_0^ta_s^2ds.$ Then we know that, $X_t|T_t$ and $Y_t|T_t$ are independent  $N(-\half T_t, T_t)$ normal random variables, and hence for all $(x,y)\in\mathbb{R}^2$,
\bq
\ex(\delta(X_t-x,Y_t-y)|a_t,T_t) ~ \frac{1}{dxdy}\mathbb{P}(X_t\in dx, Y_t\in dy|a_t,T_t)
&=&\frac{1}{dxdy}\mathbb{P}(X_t\in dx|a_t,T_t)\mathbb{P}(Y_t\in dy|a_t,T_t)\nn\\
&\le&\frac{1}{dx} \mathbb{P}(X_t\in dx|a_t,T_t)\frac{1}{\sqrt{2\pi T_t}}\nn\\
&=&\frac{1}{\sqrt{2\pi T_t}}\,\ex(\delta(X_t-x)|a_t,T_t).
\eq
Thus,
 we have
\bq
0\le\mc{J}_2&:=&\int_{-\infty}^{-R}\frac{[e^{2x}+e^{2y_K(x)}]}{e^{y_K(x)}}\mathbb{E}(a_t^2\delta(X_t-x,Y_t-y_K(x))dx\nn\\
&\le &\frac{e^{-2R}+K^2}{K-e^{-R}}\int_{-\infty}^{-R}\ex(a_t^2\delta(X_t-x,Y_t-y_K(x)))dx\nn\\
&\le &\frac{e^{-2R}+K^2}{K-e^{-R}}\frac{1}{\sqrt{2\pi}}\int_{-\infty}^{-R}\ex(a_t^2T_t^{-\half}\delta(X_t-x))dx\nn\\
&= &\frac{e^{-2R}+K^2}{K-e^{-R}}\frac{1}{\sqrt{2\pi}}\ex(a_t^2T_t^{-\half}1_{X_t\le -R}).\label{eq:29_}
\eq
By the Cauchy-Schwarz inequality we have
\bq
\label{eq:CSIneq}
\ex(a_t^2T_t^{-\half}1_{X_t\le -R})
&\le&\sqrt{\ex(a_t^{\frac{5}{2}}1_{X_t\le -R})}\sqrt{\ex(a_t^{\frac{3}{2}}T_t^{-1})}\nn\\
&\le&\sqrt[4]{\ex(a_t^{5})}\sqrt[4]{\mathbb{P}(X_t\le -R)}\sqrt{\ex(a_t^{\frac{3}{2}}T_t^{-1})}\nn\\
&\le &a_0^{\frac{5}{4}}e^{\frac{5}{2}t}\cdot\frac{e^{\frac{\pi^2}{4t}-\frac{t}{16}}}{\sqrt[4]{2\pi t}}
\sqrt[4]{\mathbb{P}(X_t\le -R)},
\eq
where the last line is due to the fact that $a_t^5=a_0^5\exp(5W_t^3-\frac{5}{2}t)$ and the bound for $\ex(a_t^{\frac{3}{2}}T_t^{-1})$ in Appendix \ref{proofexp}.

\sk

We now bound the probability $\mathbb{P}(X_t\le -R)$ by choosing an appropriate $R>0$. To that end, let $c\equiv 8\varphi(k)+{\pi^2}>0$. For fixed $a_0>0$ and $x_1<0$, it can be easily verified that
\be
\inf_{a>0}\{1+\frac{x_1^2+(a-a_0)^2}{2a_0a}\}=\inf_{a>0}\{\frac{x_1^2+a_0^2+a^2}{2a_0a}\}=\frac{1}{2a_0}\inf_{a>0}\{\frac{x_1^2+a_0^2}{a}+a\}=\frac{1}{2a_0}\cdot2\sqrt{x_1^2+a_0^2}=\sqrt{1+x_1^2/a_0^2} \nn \,.
\ee
Thus, by choosing $x_1=-a_0\sqrt{\cosh^2\sqrt{2c}-1}=-a_0\sinh\sqrt{2c}<0$, from \eqref{eq:dist} (with $y=y_0$) we see that $\sqrt{2c}$ is the minimum distance from point $(0,a_0)$ to the vertical line $x=x_1$ in $\mathbb{H}^2=\{(x,a):a>0\}$.
By Theorem 4.6 of \cite{AFLZ14}, for $x_1<0$,  we know that for all $t>0$ sufficiently small,
\be
\ex(e^{X_t}-e^{x_1})^+-(1-e^{x_1})\le C(x_1)e^{-\frac{c}{t}}t^{\frac{3}{2}},
\ee
where $C(x_1)>0$ is a constant that depends on $x_1$ only, hence by put-call parity, we have
\be
\ex(e^{x_1}-e^{X_t})^+=\ex(e^{X_t}-e^{x_1})^+-(1-e^{x_1}) \le  C(x_1)e^{-\frac{c}{t}}t^{\frac{3}{2}} .\label{q2}
\ee
Finally, observe that
\bq
\ex(e^{x_1}-e^{X_t})^+ =~\ex(e^{x_1}-S_t^{(1)})^+ ~
\int_0^\infty \pr((e^{x_1}-S_t^{(1)})^+\ge u)du
&=&\int_0^\infty \pr(e^{x_1}-S_t^{(1)}\ge u)du\nn\\
&=&\int_0^{e^{x_1}} \pr(S_t^{(1)}\le e^{x_1}-u)du\nn\\
&=&\int_0^{e^{x_1}}\pr(S_t^{(1)}\le v)dv\nn\\
&\ge&\int_{\frac{1}{2}e^{x_1}}^{e^{x_1}}\pr(S_t^{(1)}\le v)dv\nn\\
&\ge&\frac{1}{2}e^{x_1}\pr(S_t^{(1)}\le \frac{1}{2}e^{x_1})\nn\\
&=&\frac{1}{2}e^{x_1}\pr(X_t\le x_1-\log 2)\label{q3} \,.
\eq
Combining \eqref{q2} and \eqref{q3}, we have
\bq
\pr(X_t\le x_1-\log 2)\le 2e^{-x_1}C(x_1) e^{-\frac{c}{t}}t^{\frac{3}{2}} \nn \,.
\eq
\bq
\ex(a_t^2T_t^{-\half}1_{\{X_t\le -R\}})
&\le &a_0^{\frac{5}{4}}e^{\frac{5}{2}t}\cdot\frac{e^{\frac{\pi^2}{4t}-\frac{t}{16}}}{\sqrt[4]{2\pi t}}
\sqrt[4]{\mathbb{P}(X_t\le -R)},
\eq
It follows that we can choose any $R\ge\log2-x_1$, then we have
\bq
\label{eq:35}
\ex(a_t^2T_t^{-\half}1_{\{X_t\le -R\}}) &\le& a_0^{\frac{5}{4}}e^{\frac{39}{16}t+\frac{\pi^2}{4t}}\frac{1}{\sqrt[4]{2\pi t}}\sqrt[4]{\pr(X_t\le -R)} \nn \\
&\le & a_0^{\frac{5}{4}}e^{\frac{39}{16}t+\frac{\pi^2}{4t}}\frac{1}{\sqrt[4]{2\pi t}}\sqrt[4]{2e^{-x_1}C(x_1) e^{-\frac{c}{t}}t^{\frac{3}{2}} }\, \nn \\
&=& \sqrt[4]{\frac{C(x_1)e^{-x_1}}{\pi}}\cdot a_0^{\frac{5}{4}} e^{\frac{39}{16}t} \,t^{\frac{1}{8}} e^{\frac{\pi^2-c}{4t}}\nn \\
&\le& \sqrt[4]{\frac{C(x_1)e^{-x_1}}{\pi}}\cdot a_0^{\frac{5}{4}} e^{\frac{39}{16}t}e^{-\frac{2\varphi(k)}{t}}t^{\frac{1}{8}}
\eq
(recall that $c\equiv 8\varphi(k)+{\pi^2}$, $x_1=-a_0\sqrt{\cosh^2\sqrt{2c}-1}$ and we have chosen $R\ge\log2-x_1$).  Thus, we have $\mc{J}_2=o(\mc{J}_1)$ as $u\to0$ because we have $2\varphi(k)$ as opposed to $\varphi(k)$ in the exponent for $\mc{J}_2$.  Similarly, for
\bq\mc{J}_3&:=&\int_{\log(K-e^{-R})}^{\log K}\frac{[e^{2x}+e^{2y_K(x)}]}{e^{y_K(x)}}\mathbb{E}(a_t^2\delta(X_t-x,Y_t-y_K(x))dx\nn\\
&=&\int_{-\infty}^{-R}\frac{[e^{2y}+e^{2x_K(y)}]}{e^{x_K(y)}}\mathbb{E}(a_t^2\delta(X_t-x_K(y),Y_t-y)dy, \nn\eq
where $x_K(y)=\log(K-e^y)$, we also have $\mc{J}_3=o(\mc{J}_1)$ as $u\to\infty$. Overall, we have
\bq
\mc{J}&=&\mc{J}_1(1+o(1))=(1+o(1))\int_{-R}^{\log(K-e^{-R})}\frac{[e^{2x}+e^{2y_K(x)}]}{e^{y_K(x)}}\mathbb{E}(a_t^2\delta(X_t-x,Y_t-y_K(x)))dx,\nn
\eq
as $t\to 0$.

\sk
\sk
\sk

\item \textbf{Final step: computing the small-time basket call option asymptotics}

\bs

We now recall that
\bq
\label{eq:TM}
\mathbb{E}(S^{(1)}_t+S^{(2)}_t-K)^{+}\,-\,\mathbb{E}(2S_0-K)^{+}
&=& \frac{1}{2} \int_0^t  \mathbb{E}( a_u^2 \,\delta(S^{(1)}_u+S^{(2)}_u-K) [(S^{(1)}_u)^2+(S^{(2)}_u)^2]) du \,.
\eq
We first deal with the case $K \ne 2e$ and we recall the following well known asymptotic relation
\bq
 \label{eq:Awkward}
  \Upsilon(k,t)\,\,\,\,:=\,\,\,\, \int_0^t  \frac{1}{\sqrt{u   }}\,e^{-\frac{k^2}{2u}} du  ~ \frac{2 e^{-\frac{k^2}{2 t}} t - k\sqrt{2\pi t}}{\sqrt{t}}\,+\,k \sqrt{2 \pi} \,\mathrm{Erf}(\frac{k}{\sqrt{2t}})    &=&\frac{2}{k^2}\,t^{\frac{3}{2}}  e^{-\frac{k^2}{2 t}}\,\,[1+O(\frac{t}{k^2})]  \,
 \eq
for $k>0$, which just follows from the well known result that $\Phi^c(z)\sim \frac{1}{z\sqrt{2 \pi}}e^{-z^2/2}[1+O(\frac{1}{z^2})]$ as $z \to \infty$.  Note that the leading order error term here is $O(\frac{t}{k^2})$ so this approximation will generally work badly if $\frac{t}{k^2}$ is not $\ll 1$, which is often the case in practice in financial applications (we will return to this point in the numerics part in section 4.3).

\sk
Now recall that
\bq
\mc{J}&:=&\mathbb{E}(a_t^2 \delta(S^{(1)}_t+S^{(2)}_t-K) [(S^{(1)}_t)^2+(S^{(2)}_t)^2] )~ \mc{J}_1\,(1+o(1)) \eq
and from \eqref{eq:JayOne} and \eqref{eq:JayOneDouble} we know that
\bq
\mc{J}_1&=&   (1+1_{k>k^*})\,\frac{e^{-y^*} (e^{2x^*}+e^{2y^*})}{\sqrt{a_0 a^*(x^*,y^*)}}\,\frac{e^{-\half (x^* +y^*)}}{\sqrt{2 \pi t} \sqrt{\Phi_{aa}(x^*,y^*,a^*(x^*,y^*))} } \,\frac{1}{\sqrt{\Psi''(x^*)}} \, \frac{\bar{H}_K(x^*)}{\sinh \bar{H}_K(x^*)} \, e^{-\frac{\varphi(k)}{t}}[1+o(1)] \nn \\ 
\label{eq:J11111111111111}
\eq
where $x^*\equiv x^*(K)=\log(\half K)$, $y_K(x):=\log(K-e^x)$ and $\varphi(k)=\frac{1}{2}[\bar{H}_K(x^*)]^2$.  Applying \eqref{eq:Awkward} to \eqref{eq:TM} using \eqref{eq:J11111111111111}, we obtain following small-time behaviour for basket call options for $K \in (2,\infty)$ with $K \ne K^*:=2 e$:
\bq
\mathbb{E}(S^{(1)}_t+S^{(2)}_t-K)^+ &=&
 \psi(k)\, t^{\frac{3}{2}}e^{-\frac{\varphi(k)}{t}}[1+o(1)] \quad \quad \quad \quad (t \to 0)
\eq
where $k=\log K$, and
\bq
\psi(k) &=&\,(1+1_{k>k^*})\,\cdot \frac{e^{-y^*} (e^{2x^*}+e^{2y^*}) }{\sqrt{a_0 a^*(x^*,y^*)}}\,\,\frac{e^{-\half (x^* +y^*)}}{ \sqrt{2\pi \Phi_{aa}(x^*,y^*,a^*(x^*,y^*))\,} }\,\frac{1}{\Psi''(x^*)} \,  \frac{1}{\bar{H}_K(x^*)\sinh \bar{H}_K(x^*)} \nn \,
\eq
When $K\in (2,K^*]$, $x^*$ and $\varphi$ simplify to $x^*(K)=\log(\half K)$ and $\varphi(k)=\half[\cosh^{-1}(\sqrt{1+\frac{2}{a_0^2}(k-\log 2)^2})]^2$.

\bs
For the special case $K=2e$, recall from \eqref{eq:J1Special} that
\bq
\mc{J}_1
&=& \,\frac{2}{\sqrt{a_0 \bar{a} }}\,\,\frac{1}{4 \pi \sqrt{\Phi_{aa}(x^*,y^*,\bar{a})} } \,\frac{\Gamma(\frac{1}{4})\, }{(\xi \bar{H})^{\frac{1}{4}} \, t^{\frac{3}{4}}} \, \frac{\bar{H}}{\sinh \bar{H}} \, e^{-\frac{\bar{H}^2}{2t}}[1+o(1)]
\eq
where $\bar{H}=\bar{H}_{K^*}(x^*)=\cosh^{-1}\sqrt{1+\frac{2}{a_0^2}}$, $\bar{a}=a^*(x^*,y^*)=\sqrt{1+\frac{2}{a_0^2}}$.  Using that
  \bq
 \label{eq:Awkward2}
\int_0^t  \frac{1}{u^{\frac{3}{4}}}\,e^{-k^2/2u} du&=&  \frac{\sqrt{k}}{2^{\frac{1}{4}}}\Gamma(-\frac{1}{4},\frac{k^2}{2t})  ~ \,\frac{2}{k^2}\,t^{\frac{5}{4}}\,e^{-\frac{k^2}{2 t}} \,  \quad  \quad  (k>0)
 \eq
we obtain
\bq
\mathbb{E}(S^{(1)}_t+S^{(2)}_t-K)^+ &=&  \frac{1}{\sqrt{a_0 \bar{a} }}\,\,\frac{1}{4 \pi \sqrt{\Phi_{aa}(x^*,y^*,\bar{a})} } \,\frac{\Gamma(\frac{1}{4})\, }{(\xi \bar{H})^{\frac{1}{4}} \, } \, \frac{1}{\bar{H} \sinh\bar{H}} \,t^{\frac{5}{4}} e^{-\frac{\bar{H}^2}{2t}}[1+o(1)] \nn
\eq
where $\bar{H}=\cosh^{-1}\sqrt{1+\frac{2}{a_0^2}}$, $\bar{a}=\sqrt{1+\frac{2}{a_0^2}}$ and $\xi=\frac{5}{12}\frac{1}{\sqrt{2a_0^2 +4}}$.
\end{itemize}
\end{proof}

\newpage


\newpage
\bs
\bs

\section{The general bi-variate SABR model: non-zero correlation}\bs

\subsection{Small-time asymptotics for basket call options}

We now consider a generalized version of the model in \eqref{eq:Model}:
\bq
\label{eq:Model2}
        \left\{
        \begin{array}{ll}
dX_t\,=\,\, -\half a_t^2 \sigma_x^2 dt + \sigma_x a_t dW^1_t \,, \\
dY_t \,\,=\,\, -\half a_t^2 \sigma_y^2  dt + \sigma_y a_t dW^2_t \,,\\
da_t \,\,\,= \,\, \al a_t dW^3_t \, \,
        \end{array}\
        \right.
\eq
where $dW^1_t dW^2_t=\rho_{xy} dt$, $dW^1_t dW^3_t=\rho_{xa}dt$, $dW^2_t dW^3_t=\rho_{yz}dt$ with $\sigma_x,\sigma_y,\al>0$,  $\rho_{xy}^2<1$, $\rho_{xa}^2<1$, $\rho_{ya}^2<1$ and $\rho_{xy}^2+\rho_{xa}^2+\rho_{ya}^2-2\rho_{xy}\rho_{xa}\rho_{ya}<1$, which ensures that the covariance matrix of the three Brownian motions is positive semi-definite, and we set $X_0=Y_0=0$ as before.  Throughout, we set $\bar{\rho}_{xy}=\sqrt{1-\rho_{xy}^2}$, $\bar{\rho}_{xa}=\sqrt{1-\rho_{xa}^2}$ and $\bar{\rho}_{ya}=\sqrt{1-\rho_{ya}^2}$.

\sk

\begin{rem}
A fully general model would be four-dimensional, with one volatility process for each asset, but the analysis for such a model is significantly messier.  The one drawback of our existing model is that both assets have the same vol-of-vol $\al$, which may not be unreasonable if both assets are in the same sector.
\end{rem}

\sk

The heat equation associated with \eqref{eq:Model2} is given by
\bq
\p_t u\,-\,\half a^2 (\sigma_x^2 u_x+\sigma_x^2 u_y) \,+\, \half a^2 (\sigma_x^2 u_{xx} +\sigma_y^2 u_{yy}+\al^2 u_{aa}  +2\rho_{xy}\sigma_x \sigma_y u_{xy}  + 2\rho_{xa}\sigma_x \al u_{xa} + 2\rho_{ya}\sigma_y \al   u_{ya})   \nn \,.
\eq
Now let $\tau=\al^2 t$, so $\p_t u=\al^2 \p_{\tau}u$.  Then this equation transforms to
\bq
\al^2 \p_{\tau} u  \,-\,\half a^2 (\sigma_x^2 u_x+\sigma_y^2 u_y) \,+\, \half a^2 (\sigma_x^2 u_{xx} +\sigma_y^2 u_{yy}+\al^2 u_{aa}  +2\rho_{xy}\sigma_x \sigma_y u_{xy}  + 2\rho_{xa}\sigma_x \al u_{xa} + 2\rho_{ya}\sigma_y \al   u_{ya})   \nn \,.
\eq
or equivalently
\bq
 \p_{\tau} u    \,-\,\half \frac{a^2}{\al^2}(\sigma_x^2 u_x+\sigma_y^2 u_y) \,+\, \half a^2 (\frac{\sigma_x^2}{\al^2} u_{xx} +\frac{\sigma_y^2}{\al^2} u_{yy}+  u_{aa}  +2\rho_{xy}\frac{\sigma_x \sigma_y}{\al^2} u_{xy}  + 2\rho_{xa}\frac{\sigma_x}{\al}  u_{xa} + 2\rho_{ya}\frac{\sigma_y}{\al}    u_{ya})  \nn \,.
\eq
If we now set $x'=\al x/\sigma_x,y'=\al y/\sigma_y$, then PDE further transforms to
\bq
\label{eq:Simplified}
 \p_{\tau} u  \,-\,\half \frac{a^2}{\al}(\sigma_x u_x+\sigma_y u_y)  \,+\, \half a^2 ( u_{x'x'} + u_{y'y'}+  u_{aa}  +2\rho_{xy} u_{x'y'}  + 2\rho_{xa}  u_{x'a} + 2\rho_{ya}u_{y' a})   \,.
\eq
Then the diffusion matrix for \eqref{eq:Simplified} is given by $(a_{ij})=a^2\Sigma \Sigma^T$, where
\bq
\Sigma  &=&\begin{bmatrix}
 \beta & \frac{\gamma}{\bar{\rho}_{ya}} &   \rho_{xa} \nn \\
0 & \bar{\rho}_{ya} & {\rho}_{ya} \nn \\
0 & 0 & 1 \nn
        \end{bmatrix}
\eq
\nind and $\beta=\sqrt{\bar{\rho}_{xa}^2-\gamma^2/\bar{\rho}_{ya}^2}$, $\gamma=\rho_{xy}-\rho_{xa}\rho_{ya}$, and this matrix now only involves correlations.  Moreover,
\bq
\Sigma^{-1}  &=&\begin{bmatrix}
\frac{1}{\beta}    & -\frac{\gamma}{\bar{\rho}_{ya}^2 \beta } &   \frac{\xi}{\bar{\rho}_{ya}^2\beta } \\
         0& \frac{1}{\bar{\rho}_{ya} } & -\frac{\rho_{ya}}{\bar{\rho}_{ya}} \\
           0 & 0 & 1
        \end{bmatrix}
\eq
\nind where $\xi=\rho_{xy}\rho_{ya}-\rho_{xa}$.   The process associated with \eqref{eq:Simplified} now takes the form
\bq\label{38}
        \left\{
        \begin{array}{ll}
dX'_t\,=\,\, -\half a_t^2 \frac{\sigma_x}{\al} dt +  a_t (\beta dB^1_t + \frac{\gamma}{\bar{\rho}_{ya}}dB^2_t  + \rho_{xa}  dB^3_t) \,, \\
dY'_t \,=\,\, -\half a_t^2 \frac{\sigma_y}{\al}   dt +  a_t (\bar{\rho}_{ya} dB^2_t + \rho_{ya}dB^3_t)   \,,\\
da_t \,\,\,= \,\,  a_t dB^3_t \,\, \,
        \end{array}\
        \right.
\eq
\nind with $(X'_0,Y'_0,a_0)=(\frac{\al}{\sigma_x}X_0,\frac{\al}{\sigma_y}Y_0,a_0)=(0,0,a_0)$, where $B^1,B^2,B^3$ are independent Brownians.

\bs

We define the following two quantities which will be needed in the theorem which follows.
\bq
\Phi(a)&=&\half d(x,y,a)^2 \nn \\
\Psi(x)&=&\frac{d(x,y_K(x),a^*(x,y))^2}{2\al^2} \nn \\
a^*(x,y)&=& \mathrm{argmin}_{a>0}d(x,y,a) =[\frac{y^2 \al^2 \beta^2 (\bar{\rho}_{ya}^2 + \beta^2 \gamma^2) \sigma_x^2-2xy\al^2\beta^2\gamma\bar{\rho}_{ya}^2 \sigma_x \sigma_y + \bar{\rho}_{ya}^4(x^2 \al^2+a_0^2 \beta^2 \sigma_x^2) \sigma_y^2}{(\xi^2 +\beta^2 \bar{\rho}_{ya}^2)  \sigma_x^2 \sigma_y^2}]^{\half} \nn \\
\hat{A}(x,y,a) &=&-\half (\frac{\sigma_x}{\beta \al}-\frac{\gamma \sigma_y}{\bar{\rho}_{ya}^2\beta \al})x -\frac{\sigma_y}{2\al \bar{\rho}_{ya}} y  \,+\, \half \log \frac{a}{a_0} \, \nn \,.
\eq

\sk

\begin{thm} \label{thm:GeneralCase} For the general correlated model in \eqref{eq:Model2} with $X_0=0,Y_0=0$, we have the following small-time behaviour for a basket call option for $K>2$:
\bq
\mathbb{E}(S^{(1)}_t+S^{(2)}_t-K)^+ &=&
 \frac{\bar{\psi}(k)}{\sqrt{2\pi}}\, t^{\frac{3}{2}}e^{-\frac{\bar{\Lambda}(k)^2}{2 t}}[1+o(1)] \quad \quad \,\, (\text{if} \,\,\,\, \Psi''(x_j^*) > 0  \quad \forall j=1...N)
 \label{eq:BasketCallApproxCorr}
\eq
where $k=\log K$, $\bar{H}_K(x,y)=d(x,y,a^*(x,y))$, $\Lambda(k)=\min_x \bar{H}_K(x,y_K(x))$, $\bar{\Lambda}(k)=\Lambda(k)/\al^2$, $x^*_j=\mathrm{argmin}_x \bar{H}_K(x,y_K(x))$ for $j=1..N$ where $N<\infty$ is the number of global minimizers of $\bar{H}_K(x,y_K(x))$, $y_K(x)=\log(K-e^x)$ as before, $y^*_j=y_K(x^*_j)$ and
\bq
d(x,y,a) &=& \rho(\Sigma^{-1}(0,0,a_0)^{\mathrm{T}},\Sigma^{-1}(\frac{\al x}{\sigma_x},\frac{\al y}{\sigma_y},a)^{\mathrm{T}}) \, ,\nn \\
\psi(k) &=&  \sum_{j=1}^N (a^*_j)^2\,\frac{e^{x^*_j} (\sigma_x^2  e^{2x^*_j}+\sigma_y^2e^{2y^*_j})}{e^{x^*_j+y^*_j}}  \, \frac{\sqrt{g}(\Sigma^{-1}(\frac{\al x^*_j}{\sigma_x},\frac{\al y^*_j}{\sigma_y},a^*_j)^{\mathrm{T}})\,\, \chi(x^*_j,y^*_j,a^*_j)}{\sqrt{2 \pi \al  } \sqrt{\Phi_{aa}(a^*_j)\Psi''(x^*_j)}}\frac{{\al^4}/{(\sigma_x \sigma_y)}}{d(x^*_j,y^*_j,a^*_j) \sinh d(x^*_j,y^*_j,a^*_j)}\frac{1}{\mathrm{det}\Sigma} \,,\nn
\eq
$\chi(x,y,a)=e^{\hat{A}(\Sigma^{-1}(0,0,a_0)^{\mathrm{T}}, \Sigma^{-1}(\frac{\al x}{\sigma_x},\frac{\al y}{\sigma_y},a)^{\mathrm{T}})}$, $\bar{\psi}(k)=\sqrt{2\pi}\,\psi(k)$ and $\rho(.,.,.)$ is defined as in \eqref{eq:dist}.
\end{thm}

\sk

\sk

\nind \begin{proof} (of Theorem \ref{thm:GeneralCase}).
Let
\bq
\label{eq:TransformedSDEs}
\begin{bmatrix}
d\hat{X}_t    \\
         d\hat{Y}_t \\
           d\hat{a}_t
        \end{bmatrix}  \,\,=\,\, \Sigma^{-1}\begin{bmatrix}
dX'_t    \\
         dY'_t \\
           da_t
        \end{bmatrix} \,\,=\,\,
a_t^2 \Sigma^{-1}\begin{bmatrix}
 -\half  \sigma_x/\al    \\
      -\half  \sigma_y/\al    \\
           0
        \end{bmatrix}dt \,+\, a_t\begin{bmatrix}
dB^1_t \\
dB^2_t    \\
dB^3_t
         \, \end{bmatrix}
\,\,=\,\, a_t^2 \begin{bmatrix}
\frac{1}{\beta}    & -\frac{\gamma}{\bar{\rho}_{ya}^2 \beta } &   \frac{\xi}{\bar{\rho}_{ya}^2\beta } \\
         0& \frac{1}{\bar{\rho}_{ya} } & -\frac{\rho_{ya}}{\bar{\rho}_{ya}} \\
           0 & 0 & 1
        \end{bmatrix} \begin{bmatrix}
 -\half  \frac{\sigma_x}{\al}    \\
      -\half  \frac{\sigma_y}{\al}    \\
           0
        \end{bmatrix}dt +a_t\begin{bmatrix}
dB^1_t \\
dB^2_t    \\
dB^3_t
         \, \end{bmatrix}\,\,\,\,\,\,\,\,
 \eq
 with $(\hat{X}_0,\hat{Y}_0,\hat{a}_0)^T=\Sigma^{-1}(0,0,a_0)^T$.  Then from Theorem \ref{thm:Bellaiche}, we know that the joint density of $(\hat{X}_t,\hat{Y}_t,\hat{a}_t)$ behaves like
\bq
 \hat{p}(\hat{x},\hat{y},\hat{a},t)  &=& \sqrt{g}(\hat{x},\hat{y},\hat{a})\,\, e^{\hat{A}(\hat{x},\hat{y},\hat{a})}\frac{1}{(2 \pi t)^{\frac{3}{2}}}\frac{\rho}{\sinh \rho} \, e^{-\frac{\rho^2}{2t}}[1+o(1)] \quad \quad \quad \quad (t \to 0) \nn \,
\eq
where $\rho=\rho(0,0,\hat{a}_0;\hat{x},\hat{y},\hat{a})$ and $\hat{A}(x,y,a)=\int_0^1 \langle \mathcal{A},\dot{\gamma}(s) \rangle ds$, where $\gamma$ is the unique distance-minimizing geodesic joining $(0,0,\hat{a}_0)$ to $(\hat{x},\hat{y},\hat{a})$ under the metric $ds^2=\frac{1}{a^2}(dx^2+dy^2+da^2)$, and $\mc{A}$ has to be computed as before from the drift coefficient in \eqref{eq:TransformedSDEs} using the formula $\mathcal{A}^i=b^i-\frac{1}{2}\sum_j \frac{1}{\sqrt{g}}\partial_j (\sqrt{g}\,g^{ij})$; in this case we find that $\mc{A}=(a^2(-\frac{\sigma_x}{2 \beta \al} + \frac{\gamma \sigma_y}{2 \bar{\rho}_{ya}^2\beta \al}),-a^2\frac{\sigma_y}{2\al \bar{\rho}_{ya}},\half a)$ and we have
\bq
\hat{A}(\hat{x},\hat{y},a) ~  \int_0^1 \langle \mathcal{A},\dot{\gamma}\rangle \,dt &=& \int_{\gamma} [\frac{1}{a^2}\mathcal{A}^1  \frac{d\hat{x}}{dt}  \,+ \, \frac{1}{a^2}\mathcal{A}^2 \frac{d\hat{y}}{dt}  + \frac{1}{a^2}\mathcal{A}^3 \frac{da}{dt}] dt \nn \\
&=& \int_{\gamma} [(-\frac{\sigma_x}{2 \beta \al} + \frac{\gamma \sigma_y}{2 \bar{\rho}_{ya}^2\beta \al}) \frac{d\hat{x}}{dt}  \,+\, (-\frac{\sigma_y}{2\al \bar{\rho}_{ya}}) \frac{d\hat{y}}{dt}  + \frac{1}{2 a} \frac{da}{dt}] dt \nn \\
&=& -\half (\frac{\sigma_x}{\beta \al} - \frac{\gamma \sigma_y}{\bar{\rho}_{ya}^2\beta \al})\hat{x} -\frac{\sigma_y}{2\al \bar{\rho}_{ya}} \hat{y}  \,+\, \half \log \frac{a}{a_0}  \,. \nn
\eq
Then the density of $(X'_t,Y'_t,a_t)$ satisfies
\bq
 \hat{p}(x',y',a,t)  &=& \sqrt{g}(\hat{x},\hat{y},\hat{a})\,\, e^{\hat{A}(\hat{x},\hat{y},\hat{a})}\frac{1}{(2 \pi t)^{\frac{3}{2}}}\frac{\rho}{\sinh \rho} \, e^{-\frac{\rho^2}{2t}} \frac{1}{\mathrm{det}\Sigma}\,[1+o(1)] \quad \quad \quad \quad (t \to 0) \nn \,
\eq
where now $(\hat{x},\hat{y},\hat{a})^{\mathrm{T}}= \Sigma^{-1}(x',y',a)^{\mathrm{T}}$.  Transforming back to the original variables, we now obtain
\bq
 \hat{p}(x,y,a,t)  &=& \sqrt{g}(\hat{x},\hat{y},\hat{a})\,\, e^{\hat{A}(\hat{x},\hat{y},\hat{a})}\frac{1}{(2 \pi \al^2 t)^{\frac{3}{2}}}\frac{\rho}{\sinh \rho} \, e^{-\frac{\rho^2}{2\al^2 t}} \frac{1}{\mathrm{det}\Sigma}\,\frac{\al^2}{\sigma_x \sigma_y} \,[1+o(1)]  \nn \\
&=& \sqrt{g}\,(\Sigma^{-1}(\frac{\al x}{\sigma_x},\frac{\al y}{\sigma_y},a)^{\mathrm{T}})\,\, e^{\hat{A}(\Sigma^{-1}(0,0,a_0)^{\mathrm{T}}, \Sigma^{-1}(\frac{\al x}{\sigma_x},\frac{\al y}{\sigma_y},a)^{\mathrm{T}})}\frac{1}{(2 \pi \al^2 t)^{\frac{3}{2}}}\frac{d(x,y,a)}{\sinh d(x,y,a)} \, e^{-\frac{d(x,y,a)^2}{2\al^2 t}} \frac{1}{\mathrm{det}\Sigma}\,\frac{\al^2}{\sigma_x \sigma_y} \,[1+o(1)]  \nn \,
\eq
\nind as $t \to 0$, where $d(x,y,a)$ is defined in the statement of the theorem and we are using that $(\hat{x},\hat{y},\hat{a})^{\mathrm{T}}= \Sigma^{-1}(\frac{\al x}{\sigma_x},\frac{\al y}{\sigma_y},a)^{\mathrm{T}}$.  For for future reference we define
\bq
\label{eq:phat1}
\hat{p}^1(x_0,y_0,a_0;x,y,a,t) &:=&\sqrt{g}\,(\Sigma^{-1}(\frac{\al x}{\sigma_x},\frac{\al y}{\sigma_y},a)^{\mathrm{T}})\,\, e^{\hat{A}(\Sigma^{-1}(0,0,a_0)^{\mathrm{T}}, \Sigma^{-1}(\frac{\al x}{\sigma_x},\frac{\al y}{\sigma_y},a)^{\mathrm{T}})}\frac{1}{(2 \pi \al^2 t)^{\frac{3}{2}}}\frac{d(x,y,a)}{\sinh d(x,y,a)} \, e^{-\frac{d(x,y,a)^2}{2\al^2 t}} \frac{1}{\mathrm{det}\Sigma}\nn \\
\eq
to be the leading order approximation here.

 From a formal application of Laplace's method (see below for discussion on justifying this rigorously), we have
\bq
\mathbb{E}(a_t^2 \delta(X_t-x,Y_t-y))
&=& \int_{a=0}^{\infty} a^2  \hat{p}(x,y,a,t) da \nn \\
&=& (a^*)^2 \sqrt{g}\,(\Sigma^{-1}(\frac{\al x}{\sigma_x},\frac{\al y}{\sigma_y},a^*)^{\mathrm{T}})\, \frac{\chi(x,y,a^*)}{2 \pi \al  t \sqrt{\Phi_{aa}(a^*)}}\frac{d(x,y,a^*)}{\sinh d(x,y,a^*)} \, e^{-\frac{d(x,y,a^*)^2}{2\al^2 t}} \frac{1}{\mathrm{det}\Sigma}\,\frac{\al^2}{\sigma_x \sigma_y} \,[1+o(1)] \, \nn
\eq
\nind where $a^*=a^*(x,y)$ is the unique minimizer of $d(x,y,a)$.  Then applying a convolution as before, we see that
\bq
&& \mathbb{E}(a_t^2 \delta(S^{(1)}_t+S^{(2)}_t-K) [(S^{(1)}_t)^2+(S^{(2)}_t)^2]) \nn \\
&\sim& \int_{e^x=0}^{K}  \frac{e^{2x}+e^{2y_K}}{e^{x+y_K}} (a^*)^2 \sqrt{g}\,(\Sigma^{-1}(\frac{\al x}{\sigma_x},\frac{\al y}{\sigma_y},a^*)^{\mathrm{T}})\,\,  \ \frac{\chi(x,y,a)}{2 \pi \al  t \sqrt{\Phi_{aa}(a^*)}}\frac{d(x,y,a^*)}{\sinh d(x,y,a^*)} \, e^{-\frac{d(x,y,a^*)^2}{2\al^2 t}} \frac{1}{\mathrm{det}\Sigma}\,\frac{\al^2}{\sigma_x \sigma_y} \,\, d(e^x) \nn \\
&\sim& \int_{-\infty}^{k}\frac{[e^{2x}+e^{2y_K(x)}]}{e^{y_K(x)}} \,(a^*)^2 \sqrt{g}\,(\Sigma^{-1}(\frac{\al x}{\sigma_x},\frac{\al y_K(x)}{\sigma_y},a^*)^{\mathrm{T}}) \frac{\chi(x,y_K(x),a)}{2 \pi \al  t \sqrt{\Phi_{aa}(a^*)}}\frac{d(x,y_K(x),a^*)}{\sinh d(x,y_K(x),a^*)} e^{-\frac{d(x,y_K(x),a^*)^2}{2\al^2 t}}\frac{\al^2 dx}{\sigma_x \sigma_y \mathrm{det}\Sigma}  \nn
\eq
where $y_K=y_K(x)=\log(K-e^x)$.  Since $\bar{H}_K(x, y_K(x))$ is real analytic in $x\in(-\infty,\log K)$, there can only exists finitely many roots to the equation $\frac{d}{dx}\bar{H}_K(x, y_K(x))=0$ over this domain. It follows that there can only be a finite number of minimizers $x_j^*$'s such that $\bar{H}_K(x_j^*, y_K(x_j^*))=\Lambda(k)=\min_x\bar{H}_K(x, y_K(x))$.  Applying Laplace's method again, we can now re-write this expression as
\bq
\sim& \sum_{j=1}^N\frac{e^{x^*_j} [e^{2x^*_j}+e^{2y^*_j}]}{e^{x^*_j+y^*_j}} \,(a^*_j)^2 \sqrt{g}\,(\Sigma^{-1}(\frac{\al x^*_j}{\sigma_x},\frac{\al y^*_j}{\sigma_y},a^*_j)^{\mathrm{T}}) \frac{\chi(x^*_j,y^*_j,a^*_j)}{\sqrt{2 \pi \al  t} \sqrt{\Phi_{aa}(a^*_j)\Psi''(x^*_j)}}\frac{d(x^*_j,y^*_j,a^*_j)}{\sinh d(x^*_j,y^*_j,a^*_j)} \, e^{-\frac{d(x^*_j,y^*_j,a^*_j)^2}{2\al^2 t}} \frac{\al^2}{\sigma_x \sigma_y \mathrm{det}\Sigma}
\label{eq:Wazaa} \,
\eq
where $a_j^*$ now refers to $a^*(x^*_j,y^*_j)$.  Similar to before, we know from the generalized It\^{o} formula that
\bq
\label{eq:TanakaMeyer}
\mathbb{E}(S^{(1)}_t+S^{(2)}_t-K)^{+}\,-\,\mathbb{E}(2S_0-K)^{+}
&=& \frac{1}{2} \int_0^t  \mathbb{E}( (\sigma_x^2 +\sigma_y^2) a_u^2 \,\delta(S^{(1)}_u+S^{(2)}_u-K) [(S^{(1)}_u)^2+(S^{(2)}_u)^2]) du
\eq

Combining this with \eqref{eq:Awkward}, we find that the asymptotic basket call price is given by
 \bq
\sum_{j=1}^N \frac{e^{x^*_j} [\sigma_x^2 e^{2x^*_j}+\sigma_y^2 e^{2y^*_j}]}{e^{x^*_j+y^*_j}} \,(a^*_j)^2 \sqrt{g}\,(\Sigma^{-1}(\frac{\al x^*_j}{\sigma_x},\frac{\al y^*_j}{\sigma_y},a^*_j)^{\mathrm{T}}) \frac{\chi(x^*_j,y^*_j,a^*_j)}{\sqrt{2 \pi \al  } \sqrt{\Phi_{aa}(a^*_j)\Psi''(x^*_j)}}\frac{t^{\frac{3}{2}}\,e^{-\frac{d(x^*_j,y^*_j,a^*_j)^2 }{2\al^2 t}}\al^4/(\sigma_x \sigma_y) }{d(x^*_j,y^*_j,a^*_j) \sinh d(x^*_j,y^*_j,a^*_j)} \,  \frac{1}{\mathrm{det}\Sigma}\,\,[1+o(1)] \,   \,.\nn
\eq

The map from $(\hat{x},\hat{y},\hat{a})$ to $(x,y,a)$ is an invertible linear mapping, which maps  compact domains to compact domains.  Thus, the tail integral (an integer over the complement of a compact set) in the convolution of $\hat{p}_t(x,y_K(x),a^*(x,y_K(x)))$ in $x$ also corresponds to a tail integral of relevant joint density in $\hat{x}$, which can be controlled using similar arguments as in the proof of Theorem \ref{thm:SABRrhozero}.

In particular, from \eqref{38} we know that, conditional on the natural filtration generated by $(a_t)_{t\ge0}$, $\mathcal{F}_t^a$, we have
\bq
X_t'&=&-\half\frac{\sigma_x}{\al}T_t+\int_0^ta_s(\beta dB_s^1+\frac{\gamma}{\bar{\rho}_{ya}}dB_s^2)+\rho_{xa}(a_t-a_0),\nn\\
Y_t'&=&-\half\frac{\sigma_y}{\al}T_t+\bar{\rho}_{ya}\int_0^ta_sdB_s^2+\rho_{ya}(a_t-a_0).\nn
\eq
Recall that $\beta=\sqrt{\bar{\rho}_{xa}^2-\gamma^2/\bar{\rho}_{ya}^2}$, $\gamma=\rho_{xy}-\rho_{xa}\rho_{ya}$ so $\beta^2+\frac{\gamma^2}{\bar{\rho}^2_{ya}}=\bar{\rho}_{xa}^2$.
It follows that
$(X'_t|\mathcal{F}_t^a)\sim N(-\half\frac{\sigma_x}{\al}T_t+\rho_{xa}(a_t-a_0), \bar{\rho}_{xa}^2T_t)$,  $(Y'_t|\mathcal{F}_t^a)\sim N(-\half\frac{\sigma_y}{\al} T_t+\rho_{ya}(a_t-a_0),\bar{\rho}_{ya}^2T_t)$,  $X'_t$ and $Y'_t$ given $\mathcal{F}_t^a$ are correlated normal with correlation $\frac{\gamma}{\bar{\rho}_{xa}\bar{\rho}_{ya}}$. Hence, $(X_t|\mathcal{F}_t^a)\sim N(-\half \frac{\sigma_x^2}{\al^2}T_t+\rho_{xa}\frac{\sigma_x}{\al}(a_t-a_0),\frac{\sigma_x^2}{\al^2}\bar{\rho}_{xa}^2T_t)$ and $(Y_t|\mathcal{F}_t^a)\sim N(-\half\frac{\sigma_y^2}{\al^2} T_t+\rho_{ya}\frac{\sigma_y}{\al}(a_t-a_0),\bar{\rho}_{ya}^2\frac{\sigma_y^2}{\al^2}T_t)$,  $X_t$ and $Y_t$ are correlated with correlation $\hat{\rho}:=\frac{\gamma}{\bar{\rho}_{xa}\bar{\rho}_{ya}}$. It follows that, for all $(x,y)\in\mathbb{R}^2$, we have
\bq\ex(\delta(X_t-x,Y_t-y)|\mathcal{F}_t^a)&=&\ex(\delta(X_t-x)|\mathcal{F}_t^a)\cdot\ex(\delta(Y_t-y)|X_t=x,\mathcal{F}_t^a)\nn\\
&\le&\ex(\delta(X_t-x)|\mathcal{F}_t^a)\cdot\frac{1}{\frac{\sigma_y}{\al}\bar{\rho}_{ya}\sqrt{2\pi T_t}\sqrt{1-\hat{\rho}^2}}\nn
\eq
and in the second line we have used that $\mathrm{Var}(Y_t|X_t,\mathcal{F}_t^a)=(1-\hat{\rho}^2)\mathrm{Var}(Y_t|\mathcal{F}_t^a)$ by standard results on conditional Normal distributions. Similar to \eqref{eq:29_}, we then have
\bq
0&\le&\int_{-\infty}^{-R}\frac{[e^{2x}+e^{2y_K(x)}]}{e^{y_K(x)}}\mathbb{E}(a_t^2\delta(X_t-x,Y_t-y_K(x))dx\nn\\
&\le &\frac{e^{-2R}+K^2}{K-e^{-R}}\int_{-\infty}^{-R}\ex(a_t^2\delta(X_t-x,Y_t-y_K(x)))dx\nn\\
&\le &\frac{e^{-2R}+K^2}{K-e^{-R}}  \frac{1}{\frac{\sigma_y}{\al}\bar{\rho}_{ya}\sqrt{2\pi}\sqrt{1-\hat{\rho}^2}}\int_{-\infty}^{-R}\ex(a_t^2T_t^{-\half}\delta(X_t-x))dx\nn\\
&= &\frac{e^{-2R}+K^2}{K-e^{-R}}\frac{1}{\frac{\sigma_y}{\al}\bar{\rho}_{ya}\sqrt{2\pi}\sqrt{1-\hat{\rho}^2}}\,\ex(a_t^2T_t^{-\half}1_{\{X_t\le -R\}}) \nn
\eq
and we then proceed as in \eqref{eq:CSIneq}.
\nind \end{proof}

\sk



\subsection{Implied volatility}

\sk
\sk
We define the implied volatility of a basket call option with strike $K$ and maturity $t$ under the model in \eqref{eq:Model2} as the unique solution $\hat{\sigma}_t(k)$ to:
\bq
\mathbb{E}(S^{(1)}_t+S^{(2)}_t-K)^+ &=& C^{\mathrm{BS}}(2,K,\hat{\sigma}_t,t) \nn
\eq
where  $C^{\mathrm{BS}}(S,K,\hat{\sigma}_t,t)$ is the usual Black-Scholes call option pricing formula with zero interest rates.  It will be convenient to re-write this equation in a normalized form as
\bq
\half \mathbb{E}(S^{(1)}_t+S^{(2)}_t-K)^+ &=& C^{\mathrm{BS}}(1,e^{x_1},\hat{\sigma}_t(x_1),t) \nn
\eq
where $e^{x_1}=\half K$, and we now show the dependence of $\hat{\sigma}_t$ on $x_1$ explicitly.

\sk
\sk

\begin{cor}
For the model defined above, let $\hat{\sigma}_t(x_1)$ denote the implied volatility at maturity $t$ for strike $K=2e^{x_1}=e^k$, with $K > 2$.
Then we have the following asymptotic behaviour for $\hat{\sigma}_t(x_1)$:
\bq \label{eq:ImpliedVolExpansion}
\hat{\sigma}^2_t(x_1) &=& \hat{\sigma}_0(x_1)^2  +a(x_1) t\, \,+ \,o(t)\,,
\eq
\nind as $t \to 0$, where
\bq \label{eq:sigmaa}
\hat{\sigma}_0(x_1)\,=\, \frac{|x_1|}{\sqrt{2 \bar{\Lambda}(k)}}\,\quad ,\quad
a(x_1) \,=\, \frac{2\hat{\sigma}_0^4(x_1)}{x_1^2}\log\frac{\half \bar{\psi}(k)}{A_{\mathrm{BS}}(x,\hat{\sigma}_0(x_1))} \,.
\eq
\end{cor}
\begin{proof}
If we equate the small-time basket call expansion in Theorem \ref{thm:GeneralCase} (normalized by the effective initial stock price, which is $2$) and the small-time expansion for a standard European call option with initial stock price $1$ and strike price $e^{x_1}$ using the usual Black Scholes call option formula with volatility parameter equal to $\sqrt{\sigma^2+at}$ and maturity $t$ (see Proposition 3.4 in \cite{FJL12}) we get
\bq
\frac{\bar{\psi}(k)}{\sqrt{2\pi}}\, t^{\frac{3}{2}}e^{-\frac{\bar{\Lambda}(k)^2}{2 t}}[1+o(1)]  &=&\frac{A_{\mathrm{BS}}(x_1,\sigma)}{\sqrt{2\pi}}\, e^{\frac{1}{2}\frac{ax_1^2}{\sigma^4}}\,t^{\frac{3}{2}}\, e^{-\frac{x_1^2}{2\sigma^2t}}[1+o(1)]
\eq
where $A_{\mathrm{BS}}(x_1,\sigma)=\frac{\sigma^3}{x_1^2}e^{\half x_1}$.  Taking the log of both sides and cancelling terms we see that
\bq
-\frac{\bar{\Lambda}(k)^2}{2 t} \,+\,\log \bar{\psi}(k)\,   &=& -\frac{x_1^2}{2\sigma^2t} \,+\, \log [A_{\mathrm{BS}}(x_1,\sigma)]\, \,+\,\frac{1}{2}\frac{ax_1^2}{\sigma^4}\,+\,o(1)
\eq
\nind and equating the leading order and correction terms we obtain \eqref{eq:sigmaa}.  This equating argument is made rigorous in section 7.2 of \cite{FJL12} and is a model-independent argument.
\end{proof}

\subsection{Numerical results}
\sk
\sk

\sk Before delving into the numerics, we first recall the asymptotic relation
\bq
\Upsilon(k,t) \,\,\,\,:=\,\,\,\, \int_0^t  \frac{1}{\sqrt{u   }}\,e^{-\frac{k^2}{2u}} du&=&\frac{2}{k^2}\,t^{\frac{3}{2}}  e^{-\frac{k^2}{2 t}}\,\,[1+O(\frac{t}{k^2})]\eq
for $k>0$ as $t \to 0$ from \eqref{eq:Awkward}, and we note (again) that the error term inside the bracket is
$O(\frac{t}{k^2})$.  We apply this relation in the proofs of Theorem \ref{thm:SABRrhozero} and \ref{thm:GeneralCase} for third and final integration (i.e. the outer integral of the original triple integral, where we perform the final integration over $t$ using the Tanaka formula) and $k$ is given by $k=\mathrm{argmin}_{-\infty 0 <\log K }d(x,y_K(x),a^*(x,y))$.  If we apply the saddlepoint approximation formula in Theorem \ref{thm:GeneralCase} for basket calls which are closer to the at-the-money value of $K=2$ (and hence not unrealistically exponentially small in price)  $\frac{t}{k^2}$ is not $\ll 1$, so the approximation does not work so well.  Thus, in practice we recommend using the exact (closed-form) expression for $\Upsilon(k,t)$ given in \eqref{eq:Awkward} in terms of the Erf function for the final outer integral rather than the asymptotic result in \eqref{eq:Awkward}, but for completeness we compute numerics for both approximations, and the approximation given in Theorems \ref{thm:SABRrhozero} and \ref{thm:GeneralCase} still work very well for basket calls which are further from away from $K=2$.  This is an issue with \textit{any} small-time saddlepoint estimate for out-of-the-money call options under a stochastic volatility model, and is not specific to basket call options or this article.  We use the NMinimize command in Mathematica to perform the minimization in computing $\Lambda(k)=\min_x \bar{H}_K(x,y_K(x))$.

\sk
In the first table below, we have tabulated the ratio of the approximate price of the basket call computed numerically as a triple integral in Mathematica: $P^{\mathrm{numint}}(K):=\int_{x=-\infty}^{\infty} \int_{y=-\infty}^{\infty} \int_{a=0}^{\infty} (e^{x}+e^y-K)^+\hat{p}^1(x_0,y_0,a_0;x,y,a,t)da dy dx$ \footnote{where $\hat{p}^1$ is the leading order approximation to the true transition density given in \eqref{eq:phat1}} to the basket call saddlepoint approximation in \eqref{eq:BasketCallApproxCorr} (which we call $P^{\mathrm{saddle}}(K)$, see the first column on the table); in the second column we compute the same ratio but we replace the saddlepoint approximation with the adjusted formula where we use the exact expression for $\Upsilon(k,t)$ for the final integration (we call this $P^{\mathrm{saddle},\Upsilon}(K)$).  The parameters here are $t=.003$ (which is of the order of $1$ day), and $\sigma_x=\sigma_y=\sigma_a=1/\sqrt{10}\approx 0.316$, $\rho_{xy} = 0.01; \rho_{ya} = -.05; \rho_{xa} = 0.02$.).  As expected we see that both saddlepoint approximations do not work as well as $K$ tends to the at-the-money value of $2$ (because for $K$ values close to 2 we are in the \textit{moderate}, not the large deviations regime), but work very well for larger out-of-the money $K$-values.  The main purpose of this first table is not to show how well the approximation works in practice (because for $t=0.003$ the basket call prices here are too low to be of practical use), but rather to initially verify check that the formula is correct before applying it to more realistic scenarios, see next paragraph).
  \sk

  In the second table, we consider a more realistic maturity of $t=0.02$ (with the same parameters as above but now $\rho_{xy} = 0.01$, $\rho_{ya} = .05$ and $\rho_{xa} = 0.2$) and for the smaller strikes the basket call prices now take sensible values i.e. not astronomically small, and we see that the $P^{\mathrm{saddle},\Upsilon}(K)$ approximation still works well.  In this table we also compute the implied volatilities $\hat{\sigma}$ associated with $P^{\mathrm{numint}}(K)$ and $P^{\mathrm{saddlepoint},\Upsilon}(K)$ and the leading order implied volatility computed from the first equation in Eq \eqref{eq:sigmaa} (these numbers are plotted in Figure 5).  In the final three graphs, we plot these same three implied volatility smiles for three different sets of parameters with common maturity $t=.01$.

  \sk

  \sk
  Note that we have not used Monte Carlo simulation anywhere because the usual Wilard\cite{Wil97} conditioning trick cannot be applied in this context because there is no closed-form expression for basket calls under the Black-Scholes model. The other alternative would be to use importance sampling by changing to a measure under which the large deviations event becomes likely, but this would involve very messy calculations of the geodesics for the hyerbolic metric on $\mathbb{H}^3$ with a full correlation structure.\footnote{We found a mistake in the Monte Carlo implementation of the previous version and graph it produced was overly flattering because it was plotted in terms of implied volatility.}.

\bs

 \begin{center}
  \begin{tabular}{| l | l | l | l |}
    \hline
     $K \,(t=0.003)$  & $\frac{P^{\mathrm{numint}}(K)}{ P^{\mathrm{saddle}}(K)}$ &  $\frac{P^{\mathrm{numint}}(K)}{ P^{\mathrm{saddle},\Upsilon}(K)}$  \\ \hline
                2.1   & 0.86086  &  1.010238 \\ \hline
                2.3  & 0.98603 &   1.008826 \\ \hline
                2.5  & 1.00044 &   1.009752 \\ \hline
                2.7  &1.00138 &    1.006666  \\ \hline
                2.9  & 0.99875 &    1.002281 \\ \hline
                3.1  &0.99957 &    1.002183 \\ \hline
                3.3  & 1.00928 &    1.011362

 \\ \hline
  \end{tabular}
\end{center}

\bs

 \begin{center}
  \begin{tabular}{| l | l | l | l |l |l |l |}     \hline
      $K \, (t=0.02)$ &  $P^{\mathrm{numint}}(K)$ &   $P^{\mathrm{saddle},\Upsilon}(K)$ &  $\frac{ P^{\mathrm{saddle},\Upsilon}(K) }{ P^{\mathrm{numint}}(K)}$ & $\hat{\sigma}^{\mathrm{numint},\Upsilon}(K)$ &$\hat{\sigma}^{\mathrm{saddle},\Upsilon}(K)$ & $\hat{\sigma}^{0}(K)$ \\ \hline
      2.05 & 0.0090506&	0.00914671	&	1.010619 & 0.23862 &	0.23745	 & 0.22545\\ \hline
2.1 & 0.0020265	&0.00204746 &		1.010308 & 0.23308 &	0.23248	 & 0.22624 \\ \hline
2.15 & 0.000313645	&0.000316796	&	1.010046 & 0.23122 &	0.23085	 & 0.22709 \\ \hline
2.2 & 3.37991E-05	&3.41335E-05	&	1.009894 &0.23076 &	0.23052	 & 0.22799 \\ \hline
2.25 & 2.58363E-06	& 2.60899E-06 &	1.009816 & 0.23094 &	0.23077	 & 0.22894 \\ \hline
2.3 & 1.4354E-07	&1.44944E-07	&	1.009781 & 0.23145 &	0.23132	 & 0.22992 \\ \hline
2.35 & 5.95334E-09	&	6.01137E-09 &	1.009747 & 0.23216 &	0.23206	 & 0.23094 \\ \hline
2.4 & 1.8942E-10 	&1.91261E-10&		1.009719 & 0.23300 &	0.23291	 & 0.23198 \\ \hline
\end{tabular}
\end{center}

\begin{figure}
\centering
\includegraphics[width=155pt,height=150pt]{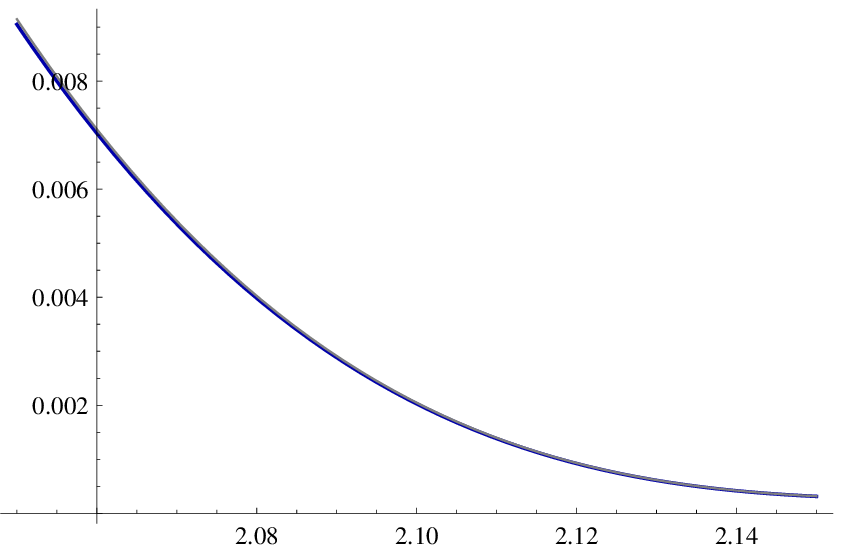}
\includegraphics[width=155pt,height=150pt]{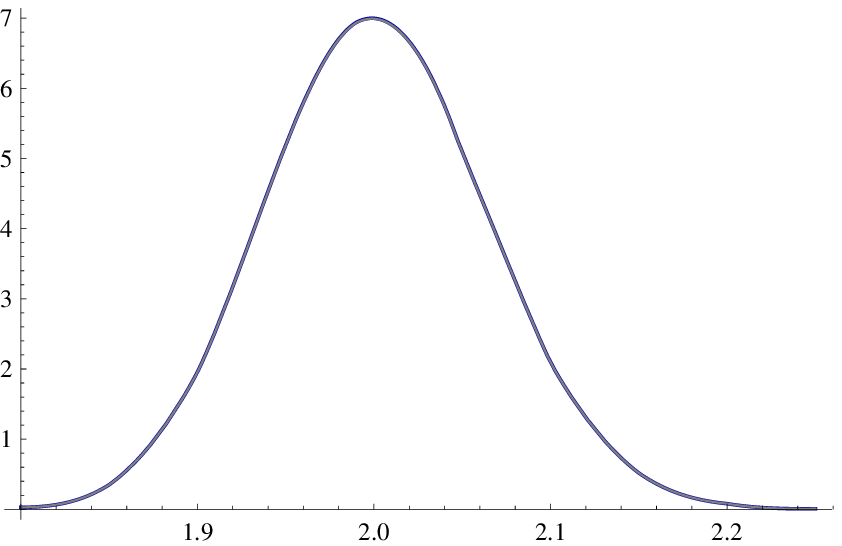}
\includegraphics[width=180pt,height=160pt]{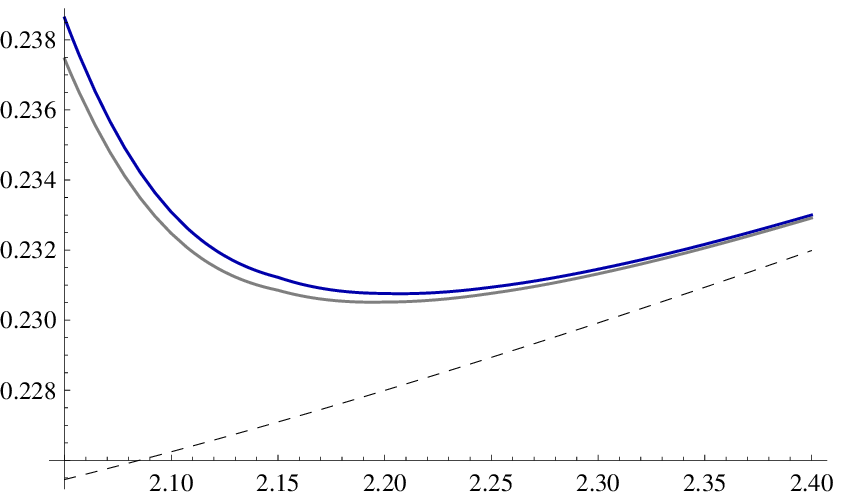}
\nind \caption{In the graph on the left, we have plotted $P^{\mathrm{saddle}}(K)$ (grey) verses $P^{\mathrm{numint}}(K)$ (blue) for $\sigma_x=\sigma_y=\sigma_a=1/\sqrt{10}\approx 0.316$, $\rho_{xy} = 0.01$, $\rho_{ya} = .05$ $\rho_{xa}= 0.2$ and $t=0.02$ (this is the data in the second table), and the parameters values given above.  In the middle plot we plot the saddlepoint approximation for the density of $S^1_t+S^2_t$ (grey) verses the density of $S^1_t+S^2_t$ via numerical integration (blue), the former is just obtained by making a trivial adjustment to the prefactors in front of the exponential in \eqref{eq:Wazaa}, similar to what we did for Figure 4 (the two density approximations are so similar here that it is difficult to make out the blue curve underneath the grey one).  In the right plot we plot the corresponding implied volatility with the same colour scheme and we also plot the leading order implied volatility $\hat{\sigma}_0(K)$ using the formula in the first equation in Eq \eqref{eq:sigmaa} (black dashed).}
\end{figure}
\begin{figure}
\centering
\includegraphics[width=180pt,height=160pt]{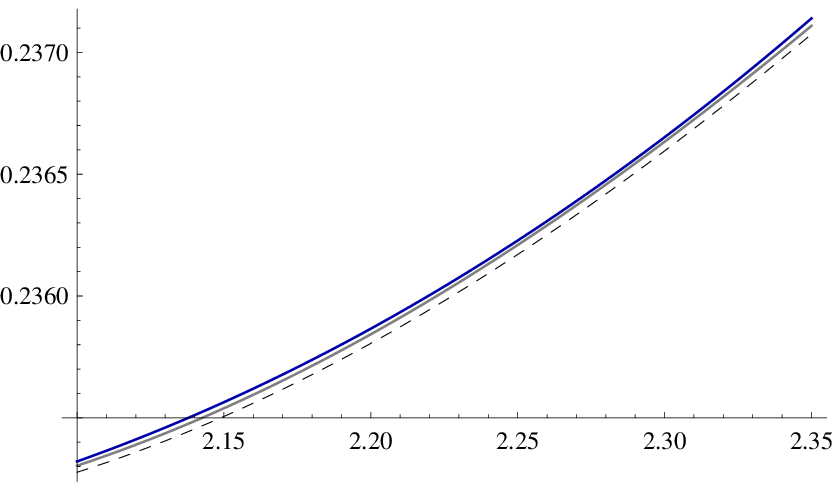}
\includegraphics[width=180pt,height=160pt]{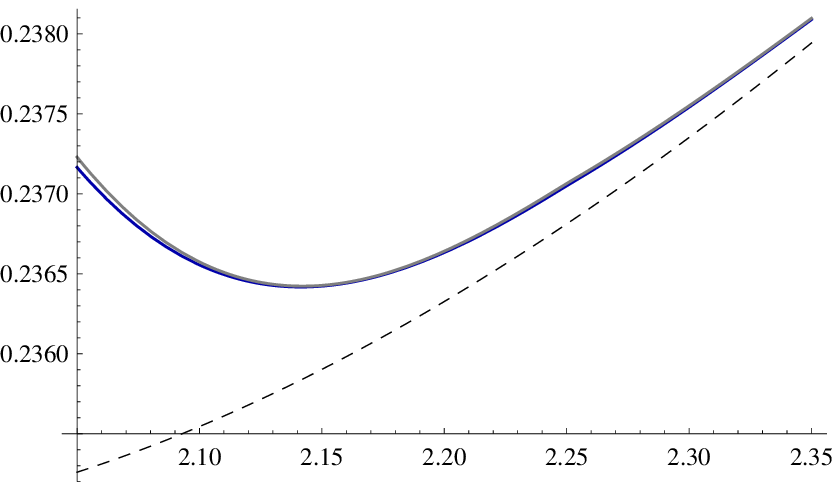}
\includegraphics[width=180pt,height=160pt]{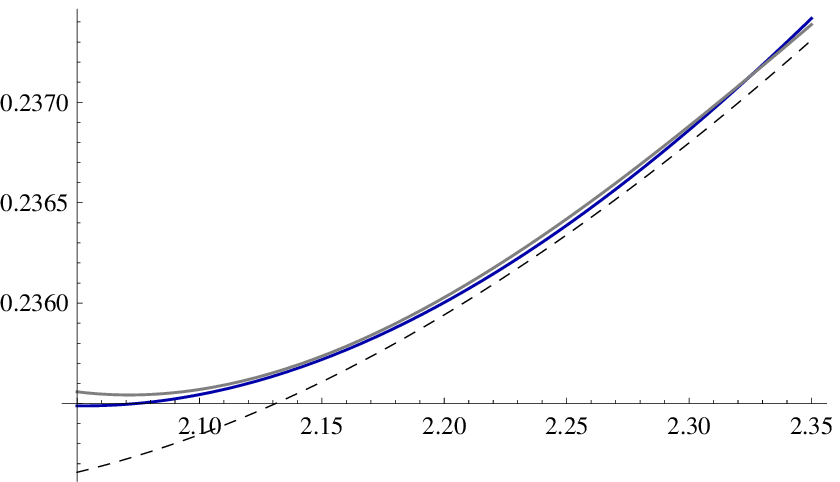}
\nind \caption{Here we have plotted the implied volatility for $P^{\mathrm{saddle}}(K)$ (grey) verses the implied volatility for $P^{\mathrm{numint}}(K)$ (blue) verses the leading order implied volatility $\hat{\sigma}_0(K)$ (black dashed) for
 $\sigma_x=1.1/\sqrt{10}$, $\sigma_y=\sigma_a=1/\sqrt{10}\approx 0.316$ and $t=.01$ in all plots with $\rho_{xy} =\rho_{ya}=\rho_{xa} =0$ (left graph), and $\rho_{xy} =0$,  $\rho_{ya}=.03$, $\rho_{xa} =0.02$ (right graph) and $\rho_{xy} =0$,  $\rho_{ya}=-.02$, $\rho_{xa} =0.03$ (lower graph)}.
\end{figure}

\subsection{Quanto and spread options}

\sk

It should be possible to adapt Theorem \ref{thm:GeneralCase} to compute small-time asymptotics for a quanto option which pays $(S^{1}_t/S^{2}_t-K)=(e^{X_t-Y_t}-K)^+$, e.g. an option on the EUR/GBP exchange rate but with payout in dollars or a spread option which pays $(S^{1}_t-S^{2}_t-K)$;  we defer the details for future research, but for the former case it is clear that $y_K(x)$ should be changed to $\log K -x$.

\bs

\appendix
\renewcommand{\theequation}{A-\arabic{equation}}
\setcounter{equation}{0}  

\bs
\section{Proof of Lemma \ref{lem:BFLError}}
\label{section:Proof}

\nind The following result will be useful:
\begin{lem}
\label{convexlem}
If $k(x)$ is $C^2$ and strictly convex over $[a,b]$, then $m(x):=k(\frac{a+b}{2}-x)+k(\frac{a+b}{2}+x)$ is strictly increasing over $[0,\frac{b-a}{2}]$.
\end{lem}
\begin{proof}
$m''(x)=k''(\frac{a+b}{2}-x)+k''(\frac{a+b}{2}+x)>0$, so $m'(x)$ is strictly increasing over $[0,\frac{b-a}{2}]$. But $m'(0)=-k'(\frac{a+b}{2})+k'(\frac{a+b}{2})=0$. The result follows.
\end{proof}

\bs

Recall that $g(z)=(\log z)^2$.  To simplify notation, we introduce
\bq u(z) \,\,\,\,:=\,\,\,\, \frac{1}{2}g''(z) \,\,\,\,=\,\,\,\,\frac{1-\log z}{z^2}\quad \quad , \quad \quad  z>0 \nn \, .\eq
Then we have $u'(z)=\frac{2\log z-3}{z^3}$, $u''(z)=\frac{11-6\log z}{z^4}$, $u(e)=0$ and $u(0+)=\infty$. Moreover, $u(z)<0$ if and only if $z>e$.

\sk

\begin{itemize}
\item  If $K\in(0,2e)$, then $u(\half K)>u(e)=0$. Since $u''(z)>0$ for all $z\in(0,e^{\frac{11}{6}})\supset(0,2e)\supset(0,K)$, we know that $u(z)$ is strictly convex over $(0,K)$. Recall that $\bar{h}_K(z)=g(z)\,+\,g(K-z)$.  Then by convexity of $u$, we know that
\bq\frac{1}{2}\bar{h}_K''(z) ~ u(z)+u(K-z)\gee 2 u(\half K) \gt 0 \,. \eq
Since $\bar{h}_K(z)$ is strictly convex and $\bar{h}_K(z)=\bar{h}_K(K-z)$, there is a unique minimizer of $\bar{h}_K(z)$ at $z^*=\half K$.

\sk

\item The case $K=2e$ is easily verified.

\sk

\item If $K\in(2e,2e^{\frac{3}{2}}]$, using that $u'(K-z)=-\frac{d}{dz}u(K-z)=\frac{-3+2\log(K-z)}{(K-z)^3}$, we see that $u(K-z)$ is strictly decreasing for $z\in(0,K-e^{\frac{3}{2}})$. Similarly, it can be easily seen that $u(z)$ is also strictly decreasing for $z\in(0,K-e^{\frac{3}{2}})\subset(0,e^{\frac{3}{2}}]$.  Thus $\frac{1}{2}\bar{h}_K''(z)=u(z)+u(K-z)$ is strictly decreasing for $z\in(0,K-e^{\frac{3}{2}})$. Furthermore, $u(z)$ is strictly convex for all $z\in[K-e^{\frac{3}{2}},e^{\frac{3}{2}}]\subset (0,e^{\frac{11}{6}})$. By Lemma \ref{convexlem}, choosing choose $a$ and $b$ such that
$(a+b)/2=\half K $ and $(b-a)/2=e^{\frac{3}{2}}-K/2$ so $[a,b]=[K-e^{\frac{3}{2}},e^{\frac{3}{2}}]$ we know that
\bq\frac{1}{2}\bar{h}_K''(\half K-y) ~ u(\half K-y)+u(\half K+y), \eq
is strictly increasing for $y\in[0,e^{\frac{3}{2}}-\half K]$. In other words, $\bar{h}_K''(z)$ is strictly decreasing for $z\in[K-e^{\frac{3}{2}},\half K]$. Overall, we have proved that $\bar{h}_K''(z)$ is strictly decreasing for $z\in(0,\half K]$, hence, there is at most one root to $\bar{h}_K''(z)=0$ over $(0,\half K)$. By $\bar{h}_K''(0+)=\infty$ and $\bar{h}_K''(\half K)<0$ we have the existence and uniqueness of the root to $\bar{h}_K''(z)=0$ over $(0,\half K)$. Denoting this root by $y_1$, then by symmetry, $\bar{h}_K(z)$ is strictly convex on $(0,y_1)$, strictly concave on $(y_1,K-y_1)$, and strictly convex on $(K-y_1,K)$. This implies that there is exactly one root to $\bar{h}_K'(z)=0$ over $(0,y_1)$. Otherwise, we would either have at least one root to $\bar{h}_K''(z)=0$ over $(0,y_1)$, or $0=\bar{h}_K'(\half K)<\bar{h}_K'(y_1)\le 0$, which is a contradiction.
Denoting this root by $z^*$, then we have two minima of $\bar{h}_K(z)$ at $z^*$ and $K-z^*$.

\sk

\item  If $K\in(2e^{\frac{3}{2}},\infty)$, then using the facts that $u(z)$ is strictly decreasing for $z\in(0,e^{\frac{3}{2}}]$, and that $u(K-z)$ is strictly decreasing for $z\in(0,K-e^{\frac{3}{2}}]\supset(0,e^{\frac{3}{2}}]$, we know that $\frac{1}{2}\bar{h}_K''(z)=u(z)+u(K-z)$ is strictly decreasing for all $z\in(0,e^{\frac{3}{2}}]\subsetneq(0,\half K)$. Moreover, since $u(e^{\frac{3}{2}})<u(e)=0$ and $u(K-e^{\frac{3}{2}})<u(e)=0$ (since $K>2e^{\frac{3}{2}}>e+e^{\frac{3}{2}}$) , we know that $\frac{1}{2}\bar{h}_K''(e^{\frac{3}{2}})<0$. But $\bar{h}_K''(0+)=\infty$, we know that there is a unique root to $\bar{h}_K''(z)=0$ over $(0,e^{\frac{3}{2}})$. Finally, for all $z\in[e^{\frac{3}{2}},\half K]$, we have $z>e$ and $K-z>e$, so  $u(z)<0$, $u(K-z)<0$, and $\frac{1}{2}\bar{h}_K''(z)<0$ for all $z\in[e^{\frac{3}{2}},\half K]$. By the same argument as the last bullet point, we can now establish the existence of two minima at $z^*\in(0,e^{\frac{3}{2}})$ and $K-z^*$.
\end{itemize}

\renewcommand{\theequation}{B-\arabic{equation}}
\setcounter{equation}{0}  
\section{Computing $\ex(a_t^{\frac{3}{2}}T_t^{-1})$}\label{proofexp}

The joint density of $(a_t,T_t)$ can be found in \cite{MY05}. It particular, $\pr(a_t\in da, dT_t\in dI)=\frac{e^{\frac{\pi^2}{2t}-\frac{t}{8}}}{\sqrt{2\pi^3t}}\frac{1}{I^2\sqrt{a}}\exp(-\frac{1+a^2}{2I})\psi_{\frac{a}{I}}(t)dadI$ where $\psi_r(t)=\int_0^\infty dz\, e^{-\frac{z^2}{2t}-r\cosh z}\sinh z\sin\frac{\pi z}{t}$ for all $r,t>0$.  From (14) of \cite{FZ14} it is known that $|\psi_r(t)|\le\frac{1}{r}$. Thus,
\bq
\ex(a_t^{\frac{3}{2}}T_t^{-1})~ \int_0^\infty\int_0^\infty a^{\frac{3}{2}}I^{-1}\pr(a_t\in da, dT_t\in dI) &\le&\frac{e^{\frac{\pi^2}{2t}-\frac{t}{8}}}{\sqrt{2\pi^3t}}\int_0^\infty\int_0^\infty\frac{a^{\frac{3}{2}}}{I}\frac{1}{I^2\sqrt{a}}\exp(-\frac{1+a^2}{2I})\frac{I}{a}dadI\nn\\
&=&\frac{e^{\frac{\pi^2}{2t}-\frac{t}{8}}}{\sqrt{2\pi^3t}}\int_0^\infty da\int_0^\infty \exp(-\frac{1+a^2}{2I})\frac{dI}{I^2}\nn\\
&=&\frac{e^{\frac{\pi^2}{2t}-\frac{t}{8}}}{\sqrt{2\pi^3t}}\int_0^\infty \frac{2}{1+a^2}da\nn\\
&=&\frac{e^{\frac{\pi^2}{2t}-\frac{t}{8}}}{\sqrt{2\pi t}}. \nn
\eq

\end{document}